\documentclass[12pt]{article} 

\usepackage{verbatim,amssymb,epsfig}  
\usepackage{amsmath, amsfonts, amstext,amsthm}
\usepackage{mathtools}
\usepackage{enumerate}
\usepackage{enumitem}
\usepackage{graphicx}
\usepackage[title]{appendix}
\usepackage{subcaption}
\usepackage{booktabs} 
\usepackage{dsfont}
\usepackage{multicol}
\usepackage{setspace}
\usepackage[flushleft]{threeparttable}
\usepackage{algorithm}
\usepackage{algpseudocode} 
\usepackage{setspace}
\usepackage{bm}
\usepackage{macros}
\usepackage[font=small]{caption}
\usepackage{url} 
\usepackage{authblk}
\usepackage{xcolor}
\usepackage{adjustbox}
\usepackage{natbib}
\usepackage{xr-hyper}
\usepackage{cleveref}
\algrenewcommand\algorithmicrequire{\textbf{Input:}}
\algrenewcommand\algorithmicensure{\textbf{Initialization:}}
\algrenewcommand\algorithmicreturn{\textbf{Return:}}
\newtheorem{lemma}{Lemma}

\theoremstyle{definition} 

\usepackage{multirow, makecell}

\newcommand{\blind}{1}

\makeatletter
\newcommand*{\addFileDependency}[1]{
  \typeout{(#1)}
  \@addtofilelist{#1}
  \IfFileExists{#1}{}{\typeout{No file #1.}}
}
\makeatother

\newcommand*{\myexternaldocument}[1]{%
    \externaldocument{#1}%
    \addFileDependency{#1.tex}%
    \addFileDependency{#1.aux}%
}

\myexternaldocument{GSBART-supp-R1} 

\newcommand*{\affmark}[1][*]{\textsuperscript{#1}}
\newcommand*{\email}[1]{#1}
\newcommand*{\affaddr}[1]{#1}

\def\cT{\calT}
\def\cM{\calM}

\setcitestyle{authoryear,open={(},close={)}} 

\mathtoolsset{showonlyrefs} 

\begin{document}

\def\spacingset#1{\renewcommand{\baselinestretch}%
{#1}\small\normalsize} \spacingset{1}


\if1\blind
{
  \title{\Large \textbf{GS-BART: Bayesian Additive Regression Trees with Graph-split Decision Rules}}
 \author{Shuren He\affmark[1], Huiyan Sang\affmark[1], Quan Zhou\affmark[1]\\
  \affaddr{\affmark[1]Department of Statistics}\\
\affaddr{Texas A\&M University, College Station} \\
\email{dri.tea, huiyan, quan@stat.tamu.edu\\}}
  \maketitle
} \fi

\if0\blind
{
  \bigskip
  \bigskip
  \bigskip
  \begin{center}
    {\Large\textbf{GS-BART: Bayesian Additive Regression Trees with Graph-split Decision Rules}}
  \end{center}
  \medskip
} \fi

\bigskip
\begin{abstract}
Ensemble decision tree methods such as XGBoost, Random Forest, and Bayesian Additive Regression Trees (BART) have gained enormous popularity in data science for their superior performance in machine learning regression and classification tasks. In this paper, we introduce a new Bayesian graph-split additive decision tree method, GS-BART, designed to enhance the performance of axis-parallel split-based BART for dependent data with graph structures. The proposed approach encodes input feature information into candidate graph sets and employs a flexible split rule that respects the graph topology when constructing decision trees. We consider a generalized nonparametric regression model using GS-BART and design a scalable informed MCMC algorithm to sample the decision trees of GS-BART. The algorithm leverages a gradient-based recursive algorithm on root directed spanning trees or chains. The superior performance of the method over conventional ensemble tree models and Gaussian process regression models is illustrated in various regression and classification tasks for spatial and network data analysis.
\end{abstract}

\noindent%
\vfill

\newpage
\spacingset{1.9} 
\section{Introduction}\label{sec:intro} 
Bayesian additive decision tree models~\citep[BART,][]{chipman2010bart} and its variants have been widely used for modeling nonparametric functions of the covariates (also called features). BART consists of multiple Bayesian decision tree weak learners where a split decision rule is used at each internal tree node to recursively partition the feature space. In the current literature, the most commonly adopted decision rule is a binary decision based on whether a selected univariate feature is higher than a certain cut point, which bipartitions the feature space in an axis-parallel fashion. When adopting the common assumption that the function value within each subspace is a constant, each Bayesian decision tree corresponds to a hyper-rectangularly shaped piecewise constant function.

The axis-parallel split rule can be inefficient in capturing non-axis-aligned decision boundaries in piecewise constant functions, as the true split decisions may depend on multiple features. Moreover, when some features lie on a complex non-Euclidean space or a network, the axis-parallel split rule ignores the domain geometry and may naively group regions separated by physical barriers into the same cluster. For example,  in spatial regression and classification problems, it is common to encounter spatial data collected from complex domains (e.g., with interior holes) or the spatial regression function having sharp changes across irregularly shaped boundaries (e.g., across a highway). Similar examples are also seen in network regression problems, where both network topology and node attributes need to be incorporated for predicting node-level responses.

Various approaches have been developed to model the decision boundary with more flexible shapes using linear combinations of multiple features \citep{blaser2016random, hada2023sparse}, but the computation of such algorithms is demanding and the linear shape can still be restrictive for modeling the decision boundary.  
Most recently, \cite{stone2024addivortes} developed a Bayesian additive model based on Voronoi tesselation clustering that yields convex hyper-polygon-shaped partitions of feature space in each weak learner. \cite{dahl2024modeling} incorporated both spatial bases and covariates as inputs of BART to model the latent intensity function for Poisson point pattern data on a road network. \cite{luo2022bamdt} considered a spatial regression setting and employed spanning tree cut models to recursively bipartition the spatial space into contiguous spatial clusters with very flexible shapes. \cite{deshpande2024flexbart} proposed a model called flexBART to handle regression problems with structured categorical predictors. They encoded the structure of categorical variables by a network and also used random spanning tree cuts to achieve more flexible partitions of the categorical data. Both \cite{luo2022bamdt} and \cite{deshpande2024flexbart} only considered spatial normal regression problems for spatially point-referenced data, and only one spatial tree graph was used for each weak learner. Moreover, they relied on the standard Bayesian backfitting algorithm that may suffer from slow mixing~\citep{pratola2016efficient}. 

In this paper, our contribution is to offer a novel regression model formulation where the feature information is encoded and incorporated into the regression function through feature graphs. Built upon feature graphs, we develop a flexible graph-split-based Bayesian additive decision tree model, termed GS-BART, as an extension of BART to accommodate complex graph-structured features. The decision rule of GS-BART is built upon bipartition models of feature graphs, thereby allowing highly flexible nonlinear decision boundaries while complying with complex domain geometries. Due to the versatility of building feature graphs, the GS-BART model framework is flexible and potentially generalizable to various settings with different types of feature data. We demonstrate its applicability using spatial regression and network regression as examples in this paper. In particular, we show that standard BART models can be unified into the GS-BART framework when encoding univariate feature information by chain graphs.

Another key contribution is to propose an efficient gradient-based informed Markov chain Monte Carlo (MCMC) algorithm, applicable to a wide range of generalized response regression problems. Standard Bayesian inference for BART uses a tailored backfitting MCMC sampler~\citep{chipman2010bart} to draw decision trees and other model parameters successively. However, the random walk split-merge Metropolis-Hastings (MH) sampler used in BART often mixes slowly \citep{ronen2022mixing, kim2025mixing}, and the marginal likelihood ratio required in the MH algorithm lacks a closed-form expression for many non-Gaussian response models.  
We propose to leverage a quadratic approximation of the loglikelihood function to explicitly express marginal likelihood ratios using the first and second-order gradient information for different response models. Further taking advantage of the order of the directed root spanning tree structure, we develop a parallel and recursive gradient message passing algorithm to quickly compute the marginal likelihood ratios for all possible neighboring states of the current decision tree. This enables a rejection-free \textsl{informed Bayesian inference} for sampling decision trees ~\citep{zhou2022rapid}. To our knowledge, informed MCMC sampling for decision trees has not been utilized in the generalized regression setting. 
We evaluate the proposed method against various competing approaches and demonstrate its superior performance across comprehensive synthetic and real-world datasets.
  
\medskip

\noindent \textbf{Related Work.}
Numerous methods have been developed for addressing regression problems with structured features. A notable class of these methods is Gaussian process regression (GPR), especially in the spatial setting, where the spatial effect is modeled through a Gaussian process, while the effects of other covariates are captured in the mean function of GPR using parametric or nonparametric regression models such as BART proposed by \cite{zhang2007spatially}, generalized additive models proposed by \citet{nandy2017additive}, random forests as explored by \citet{saha2023random} and neural networks by \cite{zhan2024neural}. However, the additive nature of the models for spatial and other covariate effects may fail to account for potential interactions between these two types of features. Moreover, traditional GPR methods often assume a parametric covariance function defined on Euclidean space, and they may struggle when the spatial random effect exhibits sharp discontinuities, heterogeneous smoothness, or lies within domains with highly irregular boundaries.

Efforts have also been made to enhance the computational efficiency of BART. For instance, \citet{murray2021log} proposed conjugate data augmentation techniques for log-linear BART models. 
\citet{he2023stochastic} introduced XBART, a stochastic search algorithm that proposes decision trees by recursively computing the marginal likelihood ratio across candidate cut points. However, XBART is biased and does not converge to the exact posterior distribution of BART. Recently, \cite{mohammadi2020continuous} 
proposed a rejection-free sampler based on a continuous-time birth-death Markov process to improve the mixing in a normal regression setting. \cite{kim2025mixing} considered an informed MCMC scheme for sampling a single decision tree with dyadic cut points. \citet{linero2024generalized} proposed a flexible generalized BART regression for non-Gaussian responses that employs the split-merge MCMC for tree updates, along with a Laplace approximation and Fisher scoring algorithm for parameter updates. Despite these advancements, these methods rely on axis-parallel split rules, which may not be suitable for data with graph relations.

\section{Preliminaries}

Let $[n_{\text{train}}]=\{1, 2, \ldots, n_{\text{train}}\}$ be the index set of $n_{\text{train}}$ training samples, $[n_{\text{test}}]=\{n_{\text{train}}+1, \ldots, n_{\text{train}}+n_{\text{test}}\}$ be the index set of $n_{\text{test}}$ testing samples. 
Define $n=n_{\text{train}} + n_\text{test}$, and we write  $[n]=\{1, 2, \dots, n\}$, denoting the indices of all samples.  
Let $\calT$ denote a decision tree that recursively bipartitions the samples into the leaf node set denoted by $\underline{\xi}(\calT)=\{\xi_1(\calT), \xi_2(\calT), \dots, \xi_{\ell}(\calT) \}$, where $\ell=\ell(\calT)$ is the number of leaf nodes of $\calT$, and $\xi_k(\calT)$ is the set consisting of the indices of samples that belong to the $k$-th leaf node. Therefore, $\underline{\xi}(\calT)$ forms a set partition of $[n]$ such that $\xi_k \neq \emptyset$ for any $k$, $\xi_k \cap \xi_{k'}=\emptyset$ for any $k \neq k'$ and $\bigcup_{k=1}^\ell \xi_k=[n]$. Given $\calT$, samples associated with each leaf node follow a leaf node-specific model, which we assume to be a normal model with an unknown constant mean (also called leaf weight).  
We let $\calM=\big(\mu_1(\calT), \mu_2(\calT), \dots, \mu_\ell(\calT)\big)$ denote the leaf weight parameters, where $\mu_k \in \bbR$  for each $k$. 

A decision tree $\calT$ and its leaf weight parameters  $\calM$ define a  partition of $[n]$ and a piecewise constant model as follows, 
\begin{equation}
    g_i( \calT,\calM) = \sum_{k=1}^{\ell(\calT)}  \mu_k(\calT) \ind_{ \xi_k(\calT) }(i),  \quad i\in [n].
\end{equation} 
Summing over $T$ decision trees, called weak learners, we have the additive decision tree,
\begin{equation}\label{eq:sumg}
\phi_i=\sum_{t=1}^T g_{t,i} (\calT_t,\calM_t)
\end{equation}

In a decision tree, the recursive bipartitions of samples are determined by a sequence of top-down decision rules at each internal decision node of $\calT$.  
In BART, the decision rules of $\calT$ depend on a set of covariates $\textbf{X}=\{X_1, \ldots, X_{p_x}\}$; at each internal node, the decision rule takes the form $X>c$, where only one covariate $X\in \bfX$ is used at a time, and $c$ is the cut point. Such a decision rule results in axis-parallel decision boundaries, which recursively partition the covariate space into $\ell$ hyper-rectangular shaped subspaces, each associated with a leaf node.  
~\cite{chipman2010bart} modeled each weak learner $\calT$ by a generative tree prior, which specifies the leaf node splitting probability, and, if a split occurs,  draws $X$ from $p_x$ features and $c$ from a set of candidate cut points to form the decision rule.  

However, features may have complex multivariate or graph structures. For example, to predict election results at the county level, in addition to a set of social demographic covariates, the network structure of the county adjacency graph might be used to construct informative decision rules to partition county-level data. Moreover, when a subset of features (e.g., latitude and longitude) is known to jointly influence the response, defining a decision rule depending on multivariate features may improve the partition efficiency. Existing BART models are too restrictive or inefficient for handling such multivariate or graph features. This motivates us to reformulate the decision tree model to allow for more flexible decision rules to accommodate complex structured features.

\section{Graph-split Bayesian Additive Decision Trees}\label{sec:additivetree}
\subsection{Model Formulation}\label{subsec:modelformu}  
We consider a generalized nonparametric regression model, called GS-BART, with the response $\bfy=(y_1, \dots, y_n)$.
Conditional on the latent variable $\phi_i=\sum_{t=1}^T g_{t,i} (\calT_t,\calM_t)$ as in \eqref{eq:sumg}, we model 
\begin{equation}\label{eq:def-likelihood}
     p(\bfy | \{ \calT_t, \calM_t\}_{t =1}^T,\sigma)=\prod_{i=1}^n f_\sigma \left( y_i | \phi_i \right)
\end{equation} 
where $f_\sigma$ is the probability density or mass function of the response model with nuisance parameter $\sigma$. For normally distributed data, the model takes the form
    $y_i \sim \phi_i + \epsilon_i$ with $\epsilon_i \iidsim \text{Normal}(0, \sigma^2)$, for $i=1,\ldots, n,$
and thus $\phi_i$ is the unknown regression mean of $y_i$. 
For Poisson count data, $\phi_i$ represents the unknown rate of $y_i$, and the model can be further generalized by assuming 
$y_i \sim \text{Poisson}(E_i \, \phi_i)$ where $E_i$ denotes the exposure variable.   

To model each decision tree $\calT$, rather than using covariates as the regression input features,  we consider a general notion of input features using a set of candidate feature graphs encoding the observed feature information. The decision rules are then based on recursive bipartitions of the selected graphs from the candidate feature graph set. We shall provide several concrete examples for constructing candidate feature graphs for different types of feature information in Section~\ref{subsec:candidategraph}. In particular, we shall show that standard BART can be unified into this flexible regression model framework by incorporating covariate information as a set of chain graphs. 
Throughout the paper, note the distinction between ``vertices" in graphs versus ``nodes" in decision trees. 

Specifically, for the $t$-th decision tree,  we assume that the input is a set of $p_t$ candidate directed rooted tree graphs (also called arborescence), denoted as $\dbmG_t= \{\dG_{j}^{(t)}\}_{j=1}^{p_t}$, encoding various feature structural information. Each candidate graph $\dG$ has its own vertex set $V$ and directed edge set $\dE$, where $V =\{v_1, v_2, \dots, v_{ |V| }\}$ is a pre-specified partition of $[n]$ such that each sample belongs to one and only one vertex.  Note that multiple samples could belong to the same vertex, and we allow some vertices to contain an empty set of samples. 
When $|V|$ is chosen to be much smaller than the sample size $n$, the vertex set construction serves as a binning strategy to reduce memory consumption and speed up model training. The edge set $\dE$ captures feature similarity structure among binned samples at each vertex. 
Each directed edge in $\dE$ points from a vertex to its unique parent vertex, except for the root vertex, which has an outdegree of $0$.  

We choose arborescences as candidate graphs due to their nice properties for recursive graph partitions; disconnecting one edge $\de$ in $\dE$ results in a bipartition of the vertex set $V$ into two connected components, $V_L$ and $V_R$, whose graphs are still arborescences. Accordingly, the samples are bipartitioned into two subsets, $S_L(\dG, \de)=\bigcup_{v \in \calV_L} v$ and $S_R(\dG, \de)=\bigcup_{v \in \calV_R} v$.  
Therefore, the structure of an arborescence supports recursive bipartitions of data which can be used as node decision rules for a decision tree. 
We will show in Section~\ref{sec:model-imp} that arborescences also enjoy many computational benefits. In particular, it naturally defines an order of the directed edges, which enables the design of an efficient informed MCMC algorithm using recursive message passing. We note that while general graphs can also reflect feature structural information, formulating and computing recursive bipartition models on such graphs is often much more cumbersome.

Given the candidate graphs set $\dbmG$, we propose a graph split decision tree, where the decision rule at each internal node is defined as a bipartition of a particular candidate arborescence, called the \emph{decision arborescence}. We refer to an edge of the decision arborescence as a \emph{decision edge}.
We adopt the following prior to specify how the graph-split rule-based decision tree is generated 
(also see Figure \ref{fig:Construct-split}(A) for an example):  
\begin{enumerate}
    \item Start with the root node representing all samples $[n]$.
    \item Split a current leaf node $\xi$ with probability $p_{\text{split}}(\xi)$. If $\xi$ splits, according to a decision rule prior $p_{\text{rule}}(\xi)$, choose one candidate graph $\dG=(V, \dE)$ from $\dbmG$ as the decision graph and one edge $\de$ from $\dE$ as the decision edge to split $\xi$ into two offspring nodes $\xi_L$ and $\xi_R$ given by 
$
    \xi_L=\xi \cap S_L(\dG, \de), \quad  \xi_R=\xi \cap S_R(\dG, \de).$
Set $\xi$ as an internal node and its two offspring nodes as the leaf nodes in the new decision tree. 
\item   Repeat Step 2 for the new decision tree.
\end{enumerate} 

\begin{figure}[h]
\centering
		\includegraphics[width=6in, height=4in]{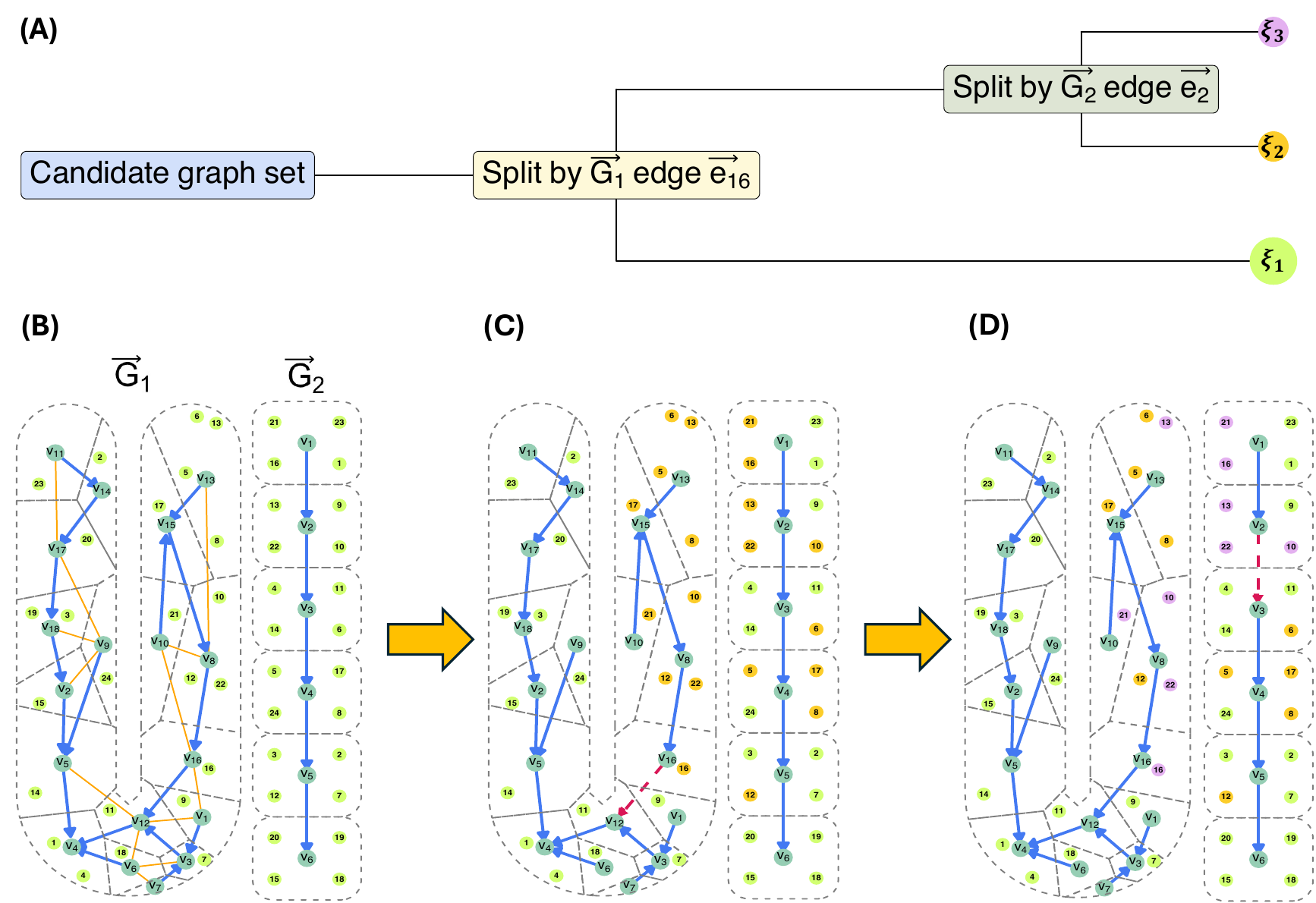}
        \setlength{\belowcaptionskip}{-5pt}
        \caption{
        (A): A graph split decision tree with two internal nodes and three leaf nodes;
        (B): The candidate graph set $\{\protect \dG_1, \protect \dG_2\}$ consisting of a spatial tree arborescence (marked by blue arrows) on a U-shape and a chain arborescence connecting bins of an unstructured numerical feature. $\protect\dG_1$ is a random arborescence from a spatial neighborhood graph $G_{0}$ (marked by both orange lines and blue arrows) connecting vertices representing Voronoi cells (marked by black dashed polygons).  
        (C)-(D): The two decision graphs and edges (marked by red dashed arrows) used in (A). The leaf node memberships for samples associated with each candidate graph vertex (marked by lime, orange, and purple colored dots) are updated as the decision tree grows.
        }   \label{fig:Construct-split}
\end{figure}

We adopt the same default prior given by \cite{chipman2010bart} for the leaf node split probability, 
$p_{\text{split}}(\xi)=\alpha(1+d_\xi)^{-\beta},$ where $\alpha \in(0,1), \beta \in[0, \infty)$ are hyperparameters, 
and $d_\xi$ is the depth of $\xi$ (the depth of the root node is 0). We recommend setting  $\alpha=0.95$, $\beta=2$ to encourage a shallow decision tree.
For the decision rule prior at a node $\xi$, we only consider the space of all \emph{valid} decision rules such that $\xi_L$ and $\xi_R$ are nonempty. 
We assume $p_{\text{rule}}(\xi)=  \pi_{e}(\de| \dG ,\xi)\pi_{g}(\dG|\xi)$, where $\pi_{g}$ is the prior distribution of $\dG\in \dbmG$ that assigns nonzero probabilities to each \emph{valid} decision graph  at node $\xi$, and similarly, $\pi_{e}$ is the prior distribution on the 
set of all \emph{valid} decision edges of $\dG$. We adopt uniform probabilities for both $\pi_{g}$ and $\pi_{e}$ in our numerical examples, but extensions to nonuniform distributions such as the sparsity-inducing method in \cite{linero2018bayesian} are straightforward. We note that the sets of valid candidate decision graphs and valid candidate decision edges vary as the decision tree grows. The details are deferred to Section~\ref{subsubsec:splitmove} and Section~\ref{app_subsec:GSB-prior}.
 
Finally, for the conditional prior of $\calM=\big(\mu_1(\calT), \mu_2(\calT), \dots, \mu_\ell(\calT)\big)$ given $\calT$, we assume $\mu_1, \mu_2, \dots, \mu_\ell$ are independent a priori, with each
$\mu_k$ following a normal prior $\pi(\mu_{k})\sim N(\mu_0,\sigma_\mu)$.
It is known that the performance of BART can be sensitive to the choice of $\sigma_{\mu}$. 
\cite{chipman2010bart} fixed $\sigma_{\mu}$ in a data-driven way such that the resulting prior variance of $\phi$ matches with the empirical data variance in the normal regression. However, the choice of $\sigma_\mu$ could be challenging beyond normal regression settings.  
Therefore, in the generalized GS-BART, we treat $\sigma_{\mu}$ as an unknown parameter similar to \citet{he2023stochastic} and learn it from data. We assign an inverse Gamma prior distribution,
$
\sigma_\mu^{2} \sim \text{inv-Gamma} \left(a/2, b/2 \right)
$. 
Based on simulation studies, we recommend $a=3$ and $b=\text{var}(\bfy)/ T$ as the default choice. 
 
\subsection{Build Candidate Graphs}\label{subsec:candidategraph} 
The choice of candidate graph set is a key ingredient in the GS-BART method; it reflects the feature structural information and governs the algorithm's computational complexity. Below, we provide several concrete examples to demonstrate how to construct candidate graphs from various types of feature information.  

\begin{example}[unstructured numerical features] 
Consider a numerical feature $X\in \bbR$ and denote its observed values for the $n$ observations by $x_1, \dots, x_n$.  
We sort the $n$ observations, and an arborescence $\dG$ then can be constructed by binning adjacent observations. Explicitly, choose a positive integer $|V| \leq n$ and cut points $-\infty = c_0 < c_1 <   \dots < c_{|V|-1} < c_{|V|} = \infty$. 
Then, $\dG$ is constructed as the chain graph with $|V|$ vertices $\{v_1, \dots, v_{|V|}\}$, where $v_k = \{i \in [n]: c_{k-1}<x_i<c_k\}$,  
and the edge set  $\dE = \{ (v_k, v_{k+1}) \colon 1 \leq k \leq |V| - 1 \}$.   
Bipartitioning the vertices set by a decision edge $(v_k, v_{k+1})$ in $\dG$ is equivalent to bipartition the $n$ observations by $X > c_k$ used in the classical BART.  
This procedure can be repeated multiple times, by using different $|V|$ or cut points, to create a set of graphs encoding the information in $X$. An example of a chain arborescence is shown in $\dG_2$ of Figure \ref{fig:Construct-split}(B).
\end{example}

\begin{example}[network features] \label{ex:network-graph}
Consider the scenario where we observe an attributed network. Each vertex of the network represents either a training sample with response and covariates or a testing sample with covariates only. The goal is to predict the response of testing sample using both the node-level covariates and the network structure. 
In addition to using chain arborescences encoding feature information from each node attribute, we can also construct candidate graphs representing network feature information. 
Let $A_0$ denote the observed network. We utilize community detection algorithms to group nodes in $A_0$ into $K$ connected subgraphs based on their network connectivity patterns~\citep{clauset2004community}. This grouping can be viewed as binning on $A_0$. We define $G_0$ as a graph where its vertex represents a bin (i.e., a subgraph) and its edge connects two bins if there exists at least one edge in $A_0$ linking the two subgraphs. This binned graph also defines the assignment of each training or testing sample to the vertex in $G_0$. 
\end{example}

In principle, a general graph $G_0$ can be recursively bipartitioned into connected components viewed as the recursive splitting of a binary decision tree. 
However, general graph bipartitions are computationally cumbersome. In contrast, tree graph partitions offer many computational advantages while still preserving feature structural information. \cite{luo2021bayesian} proved that for any arbitrary bipartition of $G_0$, there exists a subset of spanning trees of $G_0$ such that removing an edge from any tree in this subset induces the bipartition. This property has motivated the use of undirected spanning trees for graph partition problems in the literature \citep{teixeira2019bayesian,luo2021bayesian}.  
In GS-BART, we use multiple random spanning trees of $G_0$ as candidate graphs in each weak learner to enrich the graph bipartition space. In particular, 
we randomly pick a vertex as the root to orient an undirected spanning tree into an arborescence, whose order is useful for designing an efficient informed MCMC algorithm in Section~\ref{subsec:IIT}. 
Given $G_0$ with $|V|$ vertices and $|E|$ edges, algorithms such as \cite{kelner2009faster} can generate approximately uniform random arborescence samples with the computational complexity of $O(|E| \sqrt{|V|}/\delta)$, where $\delta$ is the approximation factor. Since the choice of the root vertex does not affect the graph bipartition space, for each spanning tree of $G_0$, we choose one random root to orient it.  
\begin{example}[spatial features on manifold] Consider the scenario where there are some multivariate features (e.g. latitude and longitude), denoted as $\bfs$, lying on a known manifold  $\calS$. We discretize $\mathcal{S}$ into a set of disjoint blocks, $\left\{\mathcal{S}_1^*, \ldots, \mathcal{S}_J^*\right\}$, such that $\bigcup_j \mathcal{S}_j^*=\mathcal{S}$. Common choices of $\mathcal{S}^*$ include triangular meshes, constrained Voronoi tessellations~\citep{voronoi1908nouvelles}, and regular rectangular or hexagonal meshes.
We use the constrained Voronoi tessellation method in our numerical examples. It partitions $\calS$ into blocks for a given set of reference points on $\calS$ such that each region contains all points closest to a given reference point compared to any others. Each block serves as a spatial bin, and we construct a spatial neighborhood graph $G_{0}$, where its vertex $v$ represents each block and the samples falling into it, and edges connect blocks sharing a border. 
The number of reference points, i.e., $|V|$ can be chosen to be much smaller than the number of observations to reduce the computation on graphs for large-scale problems.  

Note that $G_0$ is a general structural graph. Similar to Example~\ref{ex:network-graph}, we sample a number of random arborescences to include in our candidate graph set for each weak learner. 
Figure \ref{fig:Construct-split} (B) shows an example of using Voronoi tessellation to bin a U-shape manifold and construct the spatial structural graph $G_{0}$ and its arborescence $\dG$. A contiguous bipartition of $\dG$ results in a contiguous bipartition of $G_0$ and the U-shape manifold. Unlike the Voronoi tesselation clustering methods where each Voronoi cell is a convex hyper-polygon-shaped cluster, graph bipartition allows Voronoi cells of reference points to freely merge and form flexible contiguous space partitions while respecting the manifold domain geometry. 
 
\end{example}
 
We emphasize that GS-BART is highly versatile because of its flexibility in constructing candidate graph sets. Across the ensemble of weak learners, we allow $\dbmG_t$ to be constructed differently to enhance model diversity, reduce computation, or reflect prior model assumptions. In Example 2,  GS-BART can easily accommodate multiple observed networks by constructing random arborescences for each network and combining them into the candidate graph set.  Similarly, in Example 3, we may let the tessellation mesh, $\left\{\mathcal{S}_1^*, \ldots, \mathcal{S}_J^*\right\}$, vary in location and size across weak learners to enhance the diversity of candidate partitions for more flexible decision boundaries. 
Even with a relatively small sample of random arborescences in each weak learner for the network and spatial examples, we find that GS-BART could achieve reasonably good model performance in our numerical examples as long as the number of weak learners is large enough. More details are presented in the sensitivity analysis in Section~\ref{sec:experi}. Furthermore, letting candidate graph sets vary among weak learners allows users to incorporate prior assumptions on interaction effects among features. For instance, if identifying or including the interaction effects between specific features is important, candidate graph sets should be constructed such that arborescences representing these features appear simultaneously in the same candidate graph set.

\section{Bayesian Inference of GS-BART}\label{sec:model-imp}
\subsection{Standard Backfitting Algorithm and Its Limitations} \label{subsec}
We consider the generalized regression framework as in~\eqref{eq:def-likelihood}. We denote the likelihood and loglikelihood function given each $y_i$ by $L_{\sigma}(\phi_i) = f_\sigma(y_i | \phi_i)$ and $\mathcal{L}_{\sigma}(\phi_i) = \log f_\sigma(y_i | \phi_i)$,  respectively. We  assume that $\mathcal{L}_{\sigma}(\phi_i)$  is twice differentiable with respect to $\phi_i$, which is satisfied by many response models such as normal, Poisson, and multinomial regressions.  
 
The posterior distribution of the parameter vector $(\{\calT_t\}_{t=1}^T, \{\calM_t\}_{t=1}^T, \sigma_\mu, \sigma)$ is usually computed by the backfitting MCMC sampler, which successively draws $(\calT_1, \calM_1), \ldots, (\calT_T, \calM_T)$, $\sigma_\mu$, and $\sigma$ from their respective posterior conditional distributions.  
Sampling of $\sigma$ is response model specific and relatively straightforward, so we omit the details.

Let $\calM_t = \big(\mu_{t,1},   \dots, \mu_{t, \ell_t}\big)$, where  $\ell_t$ is the number of leaf nodes of $\cT_t$, and $\mu_{t,k}\coloneqq \mu_{k}(\calT_t)$ is the leaf weight at the $k$-th leaf node of $\calM_t$. 
Define $\vartheta_{-t} = (\{\calT_k\}_{k \neq t}, \{\calM_{k}\}_{k\neq t}, \sigma_\mu, \sigma)$. We omit data from the conditional posterior distributions below for notational simplicity.  BART algorithms sample $\calT_t$ from $p(\calT_t | \vartheta_{-t}) = \int p(\calT_t, \calM_t | \vartheta_{-t}) d\calM_t$ by integrating out $\calM_t$. 
Specifically, $p(\calT_t | \vartheta_{-t}) \propto \pi(\calT_t) \prod_{k = 1}^{\ell_t} m_{\xi_k}$, where $\pi(\calT_t)$ is the prior distribution of $\calT_t$, whose expression is given in Supplementary Section~\ref{app_subsec:GSB-prior}, and $m_{\xi_k}$ is the marginal data likelihood at the leaf node $\xi_k$ taking the form 
\begin{equation}\label{eq:marlike}
m_{\xi_k}= m_{\xi_k}(\vartheta_{-t}) = \int \pi(\mu_{t,k};\mu_0,\sigma^2_{\mu}) \prod_{i \in \xi_k \cap [n_{\text{train}}]}  L_{\sigma}(  \phi_{-t,i} + \mu_{t,k} ) d \mu_{t, k}, 
\end{equation}
where $\phi_{-t,i}=\sum_{k \neq t} g_{t,i}(\calT_t,\calM_t)$, and $\pi(\mu_{t,k})$ is the normal prior of $\mu_{t,k}$. Given $\calT_t$, we draw each $\mu_{t, k}$ from $ p(\mu_{t, k} | \calT_t,  \vartheta_{-t})\propto \pi(\mu_{t,k};\mu_0,\sigma^2_{\mu}) \prod_{i \in \xi_k \cap [n_{\text{train}}]}  L_{\sigma}( \phi_{-t,i} + \mu_{t,k} )$.

The standard BART algorithm is a sampler of the split-merge type.  
At each iteration, a decision tree $\calT$ is selected randomly, and a new tree $\calT^\star$ is proposed by performing either a split or merge move of $\calT$. In the split move, one splittable leaf node of $\calT$ is randomly chosen and split into two offspring nodes, while a pruning move does the opposite by randomly merging two leaf nodes with the same parent node. 
The resulting $\calT^\star$ is accepted or rejected following the standard MH procedure. Since arborescence enables recursive bipartitions of data through decision edges, the split and merge proposals can be adapted to the context of GS-BART. A split move involves proposing a new graph split decision rule at the selected leaf node by choosing a valid decision arborescence and a valid decision edge. 
We defer the details to Section~\ref{subsec:quadratic}.  
However, this random-walk split-merge-type proposal tends to have a low acceptance rate, resulting in slow mixing of MCMC \citep{pratola2016efficient}. Furthermore, the MH acceptance ratio involves a marginal likelihood ratio $\prod_{k =  1}^{\ell\left( \calT^{*}\right)} m_{\xi_k^{*}}/ \prod_{k =  1}^{\ell\left( \calT \right)} m_{\xi_k}$, which typically lacks a closed form for many non-Gaussian response models, and the same is true for the posterior distribution of leaf weights. 

In light of these limitations, below, we develop a rejection-free informed MCMC algorithm in Section~\ref{subsec:IIT} to improve the convergence of MCMC.  We also develop a quadratic loglikelihood approximation method in Section~\ref{subsec:quadratic} to facilitate the computations of the informed proposal distributions, which enables conjugacy to simplify the computations of marginal likelihood ratios and sampling of leaf weights.  
 
\subsection{Quadratic Approximation of Marginal Likelihood}\label{subsec:quadratic}
 
In regression models with normal errors, $\log m_{\xi}$ can be expressed as a closed quadratic form due to  conjugacy, which, unfortunately, does not hold in many generalized linear models.  We use the second-order approximation of the response data likelihood function to derive an approximation of $\log m_{\xi}$ with a closed quadratic form. 
Let $\hat{\phi}_{i}$ be an estimated value of $\phi_i$. In our numerical examples, we set $\hat{\phi}_{i}$ as the value generated in the previous MCMC iteration,  
which implies that $\phi_i-\hat{\phi}_i=\mu_{t,k}-\hat{\phi}_{t,i}$  for $i \in \xi_{k}$, where $t$ is the index of the decision tree being modified in the current iteration. 
Let $\dot{\mathcal{L}}_{\sigma}(\hat{\phi}_i)$ and $\ddot{\mathcal{L}}_{\sigma}(\hat{\phi}_i)$ denote the first and second order derivatives of $\calL_{\sigma}$ with respect to $\phi_i$, evaluated at $\hat{\phi}_{i}$. 
Under the generalized response model setting, $ \mathcal{L}_{\sigma}(\phi_{i}) \approx \hat{\mathcal{L}}_{\sigma}(\phi_{i})=\mathcal{L}_{\sigma}(\hat{\phi}_{i}) +   ( \mu_{t,k} - \hat{\phi}_{t,i} ) \dot{\mathcal{L}}_{\sigma}(\hat{\phi}_{i})  +   (\mu_{t,k} - \hat{\phi}_{t,i})^2  \ddot{\mathcal{L}}_{\sigma}(\hat{\phi}_{i}) /2 $,
which yields the following approximation of the marginal data likelihood and conditional posterior distributions of $\calM_t$ and  $\sigma_{\mu}^2$.  
\begin{lemma}\label{lm:marg-like}
Let $\hat{m}_{\xi_k}$ denote the marginal likelihood given in~\eqref{eq:marlike} with $\mathcal{L}_{\sigma}(\phi_{i})$ replaced by its second-order Taylor expansion at $\hat{\phi}_i$. Then,  
\begin{equation}\label{eq:def-m-xi}
	\begin{aligned}
		\log \hat{m}_{\xi_k} = &  - \frac{\mu_0^2}{2 \sigma_{\mu}^2} +\sum_{i \in \xi_k \cap [n_{\text{train}}] } \Bigl( \mathcal{L}_{\sigma}(\hat{\phi}_{i}) -  \hat{\phi}_{t,i}  \dot{\mathcal{L}}_{\sigma}(\hat{\phi}_{i}) +  \frac{\hat{\phi}_{t,i}^2 \ddot{\mathcal{L}}_{\sigma}(\hat{\phi}_{i})  }{2}  \Bigl) \\
		& + \dfrac{\big(J(\xi_k) + \frac{\mu_0}{\sigma_{\mu}^2}\big)^2}{2 \big(H(\xi_k) + \frac{1}{\sigma_{\mu}^2}\big)} - \log \sigma_{\mu} \sqrt{H(\xi_k) + \frac{1}{\sigma_{\mu}^2}}, 
	\end{aligned}
\end{equation}
where $J(\xi_k) = \sum_{i \in \xi_k \cap [n_{\text{train}}]}  (\dot{\mathcal{L}}_{\sigma}(\hat{\phi}_{i}) - \hat{\phi}_{t, i} \ddot{\mathcal{L}}_{\sigma}(\hat{\phi}_{i})    ) $ and $H(\xi_k) =  - \sum_{i \in \xi_k \cap [n_{\text{train}}]} \ddot{\mathcal{L}}_{\sigma} (\hat{\phi}_{i} )$. 
Under this approximation, the conditional posterior distributions of $\mu_{t,k}$ and $\sigma_{\mu}^2$ are given by
  \begin{align}\label{eq:updateM}
     \mu_{t,k} & |  \cT_t, \cT_{-t}, \cM_{-t}, \sigma_\mu, \sigma \sim \;  N \left(\frac{J(\xi_k) + \frac{\mu_0}{\sigma_{\mu}^2}}{H(\xi_k)+ \frac{1}{\sigma_{\mu}^2}},\; \frac{1}{H(\xi_k)+ \frac{1}{\sigma_{\mu}^2}} \right), \; \text{for }  k  = 1, \dots, \ell_t, \\
       \sigma_{\mu}^2 & | \{\cT_{t}, \cM_{t}\}_{t=1}^T, \sigma \sim\;  \text{inv-Gamma} \left( \frac{ \sum_{t = 1 }^T \ell_t + a}{2}, \; \frac{ \sum_{t = 1 }^T \sum_{k = 1}^{\ell_t}\left(\mu_{t,k} - \mu_0 \right)^2 + b}{2} \right).
\end{align}
\end{lemma}

\begin{proof}
See Supplementary Section~\ref{app_subsec:GSB-approxi}.     
\end{proof}

For the normal response model,~\eqref{eq:def-m-xi} holds exactly with $\dot{\mathcal{L}}_{\sigma}(\hat{\phi}_i) = (y_i-\hat{\phi}_i)/\sigma^2$, $\ddot{\mathcal{L}}_{\sigma}(\hat{\phi}_i) = - 1/\sigma^2$, $J(\xi)  = \sum_{i \in \xi \cap [n_{\text{train}}]} (y_i -  \hat{\phi}_{i} + \hat{\phi}_{t,i})/\sigma^2$ and $H(\xi)  =  | \xi \cap [n_{\text{train}}] |/\sigma^2.$ 
Based on this quadratic approximation, we can compute the log marginal likelihood ratios for the split and merge moves following the closed-forms given in \eqref{app_eq:split likelihood ratio} and \eqref{app_eq:merge likelihood ratio}. For example, when splitting a node $\xi$  into two offspring nodes $\left(\xi^*_{L}, \xi^*_{R} \right)$,  we only need to compute $\log \hat{m}_{\xi^*_{L}} +\log \hat{m}_{\xi^*_{R}}-\log \hat{m}_{\xi}$ to obtain the log marginal likelihood ratio between $\calT^*$ and $\calT$. 
  
\subsection{Informed Rejection-free MCMC for Posterior Sampling}\label{subsec:IIT}
We now introduce our Gibbs algorithm for sampling from the joint posterior distribution of $(\{\calT_t\}_{t=1}^T, \{\calM_t\}_{t=1}^T, \sigma_\mu, \sigma)$ outlined in Alg.~\ref{alg:GSB-procedure}.   
Let $N$ denote the number of Gibbs iterations. In each iteration, we use a deterministic-sweep scheme to draw posterior samples of $\{\calT_{t}, \calM_{t}\}_{t=1}^T$  sequentially.  
In particular, to efficiently sample a new $\calT_t$ from the conditional distribution $p(\calT_t | \vartheta_{-t})$, we develop an informed importance tempering (IIT) algorithm, a rejection-free MCMC scheme proposed by~\cite{zhou2022rapid} that extends the standard MH algorithm using ``informed proposals.'' It was shown in~\cite{zhou2022rapid}  that this algorithm may converge much faster than MH algorithms under certain scenarios and is especially useful for problems like variable selection.  
Fix an arbitrary $t$ and $\vartheta_t$, and let $q(\calT^{*}_t |  \calT_t )$ denote the probability of proposing a new tree $\calT^{*}_t$  given current $\calT_t$ according to the standard uninformed random-walk split-and-merge proposal (see Supplementary Section~\ref{app_subsec:GSB-prior} for details). Define the ``neighborhood" of $\calT_t$ as $\mathcal{N}(\calT_t)=\{\calT^{*}_t: q(\calT^{*}_t | \calT_t) > 0\}$, which includes all the possible split and merge proposals that can be applied to $\calT_t$. 
An informed proposal re-weights these neighboring decision trees in light of their posterior probabilities, 
and we assign to each $\calT^{*}_t \in \mathcal{N}(\calT_t)$ a proposal weight given by  
\begin{equation}\label{eq:def-weight}
	\eta(\calT^{*}_t | \calT_t, \vartheta_{-t})=q(\calT^{*}_t | \calT_t) h\left(\frac{ p(\calT^*_t  | \vartheta_{-t}) q(\calT_t | \calT^{*}_t)}{ p(\calT_t   | \vartheta_{-t})  q(\calT^{*}_t | \calT_t)}\right)
\end{equation}
where $h(\cdot)$ is the balancing function satisfying $h(x)=xh(1/x)$~\citep{zanella2020informed}, whose form is specified to guarantee the detailed balance condition of $\eta$ with respect to the conditional posterior. We adopt $h(x)=\sqrt{x}$ in this study. 
We set $\eta (\calT^{*}_t | \calT_t, \vartheta_{-t}) = 0$ if $\calT^{*}_t \notin \mathcal{N}({\calT}_t)$. 
The proposal probability of $\calT^{*}_t$, according to $\eta$, is given by $\eta(\calT^{*}_t | \calT_t, \vartheta_{-t} ) / Z(\calT_t, \vartheta_{-t})$, where 
\begin{equation}\label{eq:def-Z} 
Z(\calT_t, \vartheta_{-t})  = \sum_{\calT^{*}_t \in \mathcal{N}(\calT_t)} \eta\left(\calT^{*}_t | \calT_t, \vartheta_{-t}\right) 
\end{equation} 
is the normalizing constant. To avoid the low acceptance rate issue of MH algorithms, we always accept such an informed proposal, which results in a Markov chain with transition matrix  $\bm{P} (\calT, \calT^{*})=\eta(\calT^{*}_t | \calT_t, \vartheta_{-t} ) / Z(\calT_t, \vartheta_{-t})$ for each $\calT^{*}_t \in \mathcal{N}(\calT_t)$. 
After collecting the samples from $\bm{P}$, we correct for the bias by importance sampling, which is possible because $\bm{P}$ is reversible with respect to the distribution  $\tilde{p}(\calT_t   | \vartheta_{-t}) \propto  p(\calT_t   | \vartheta_{-t}) Z(\calT_t, \vartheta_{-t})$. Hence, we can simulate the Markov chain $\bm{P}$ using importance reweighted sampling to generate a new decision tree $\calT^*_t$ from the posterior $p(\calT_t   | \vartheta_{-t}) $; this procedure is summarized in Alg.~\ref{alg:one-weak-learner}. The details of the IIT algorithm and its validity are explained in Supplementary Section~\ref{app_subsec:GSB_val}.

\begin{minipage}{0.44\textwidth}
\begin{algorithm}[H]
    \centering
    \setstretch{0.5}
	\caption{GS-BART procedure}\label{alg:GSB-procedure}
	\begin{algorithmic} 
		{\small 
			\Require $\bfy$, $\{\dbmG_t\}_{t=1}^{T}, \dot{\mathcal{L}}_{\sigma}(\cdot), \ddot{\mathcal{L}}_{\sigma}(\cdot)$, $N$.   
                \Ensure $\{\hat{\phi}_i\}_{i=1}^{n_{\text{train}}}$.
			\For {j = 1, \ldots, N}
			\State Plug in $\left\{\hat{\phi}_i\right\}_{i = 1}^{n_{\text{train}}}$  and $\bfy$ to $\dot{\mathcal{L}}_{\sigma}(\cdot), \ddot{\mathcal{L}}_{\sigma}(\cdot)$.
			\State Initial $\calT_t$ to the root node.
			\For {t = 1, \ldots, T}
			\State Sample $\calT_t$ using IIT from the root node. 
			\State Sample $\calM_t$ from the posterior distribution $\pi(\calM_t | \calT_t, -)$ in \eqref{eq:updateM}.
   \State Update $\{\hat{\phi}_{t,i}\}_{i=1:n}$ from $g_{t}(\calT_t^{(j)}, \calM_t^{(j)})$ 
   \EndFor
			\State Update $ \hat{\phi}_i \leftarrow \sum_{t=1}^T \hat{\phi}_{t,i}$.
                \State Update $\sigma_{\mu}^2$ and $\sigma$
			\EndFor
		}
	\end{algorithmic}
\end{algorithm}
\end{minipage}
\hfill
\begin{minipage}{0.47\textwidth}
\begin{algorithm}[H]
    \setstretch{1}
    \centering
	\caption{
		IIT Algorithm to draw informed proposal samples for one weak learner in GS-BART}\label{alg:one-weak-learner}
	\begin{algorithmic} 
		{\small 
			\Require $\bfy$, $\calT_t$ with only root node, $\dbmG_t$, $\left\{\dot{\mathcal{L}}_{\sigma}(\hat{\phi}_{i}), \ddot{\mathcal{L}}_{\sigma}(\hat{\phi}_{i})\right\}_{i = 1}^{n}$, the number of informed proposal samples of ($\calT_t, \calM_t$) at each Gibbs iteration, $K$. 
			\For {$k = 1,...,K$}
			\State Calculate $\eta\left(\calT^{*}_t | \calT_t , \vartheta_{-t} \right)$  for all possible $\calT^{*}_t \in \mathcal{N}_{\calT}$ in \eqref{eq:def-weight}.
			\State Calculate $Z(\calT_t, \vartheta_{-t}) =  \sum_{\calT^{*}_t \in \mathcal{N}_{\calT_t}} \eta\left(\calT^{*}_t | \calT_t, \vartheta_{-t}\right)$.
			\State Randomly propose a $\calT^{*}_t$ w.p. $\frac{\eta(\calT^{*}_t | \calT_t)}{Z(\calT_t)}$.
			\State Set $\calT_t \leftarrow \calT^{*}_t$.
			\EndFor
		}
		\State Sample $\calM_t$ from \eqref{eq:updateM}.
	\end{algorithmic}
\end{algorithm}
\end{minipage}

\subsection{Recursive Computation of Marginal Likelihood Ratios} \label{subsubsec:splitmove}

In each iteration of Alg.~\ref{alg:one-weak-learner}, the informed proposal requires the evaluation of the weighting function $\eta$ and its normalizing constant, which involves the marginal likelihood ratios for every possible split and merge move at $\calT_t$.  
Since a split may occur at any splittable leaf node according to any valid decision edge of a candidate arborescence, the number of possible split moves increases rapidly for large datasets, and consequently, the naive computation of the marginal likelihood ratio for each possible move becomes impractical. 
Taking advantage of the ordered tree structure of arborescences, we design an efficient parallel and recursive algorithm for calculating the marginal likelihood ratio for each possible split. 
For notation simplicity, we assume the tree index $t$ is fixed and drop it from $\calT_t$ and $\dbmG_t$ in this section. 
Before presenting the algorithm details, we first define the order of edges of an arborescence and explain what constitutes a valid split move. 

\medskip 
\noindent\textbf{Order of arboresencence edges}.
Figure \ref{fig:vertex edge type}(A) shows an example of the vertex types and edge orders of a tree arborescence.
An edge $\de_v\coloneqq(v,\operatorname{pa}(v)) \in \dE$ is from a child vertex $v$ to its unique parent vertex $\operatorname{pa}(v)$, and hence uniquely determined by $v$. We call $v$ a \emph{root} vertex if it does not have a parent and denote it as $v_0$.  We call $v$ a \emph{bottom} vertex and $e_v$ a \emph{bottom} edge if $v$ does not have any children, and we use $B(\dG)$ to denote the set of bottom vertices of a graph $\dG$. Every vertex except the \emph{bottom} vertices has a non-empty set of child vertices, denoted as  $\operatorname{ch}(v)$. 
Given $\dG$, we can order the directed edges in $\dE$ in a bottom-up direction by tracing the parent of each child vertex of $\de_v$ from the bottom to root vertices. Meanwhile, by utilizing such bottom-up order, we can also define the ancestor and descendant vertices of a vertex $v$, denoted as $\operatorname{ances}(v)$ and $\operatorname{desc}(v)$, respectively. 

\medskip 
\noindent\textbf{Valid split moves}.
Given the candidate graphs set $\dbmG$ and the current decision tree $\calT$, a \emph{valid} candidate split move is defined by identifying a splittable leaf node $\xi$ and a \emph{valid} graph split decision rule consisting of a decision graph and a decision edge, ensuring that the training data at $\xi$ is bipartitioned into two offspring nodes containing nonempty training data. 
Some valid decision edges may result in the same bipartition of training data and are thus grouped as a set of equivalent decision edges. 
Within each equivalent edge set, we select one edge such as the bottom one to be valid, and all the other edges are regarded as redundant decision edges. In Supplementary Alg.~\ref{app_alg:recur split}, we demonstrate the method of determining edge type where we use `1', `-1', and `0' to denote valid, invalid, and redundant edges. We set the log marginal likelihood ratio of invalid and redundant split moves as negative infinity. 
Note the determination of edge types depends on the splittable leaf node.  The details are described in Supplementary Section \ref{app_subsec:redund-egde}. 
Figure \ref{fig:vertex edge type}(B) shows an example of how the types of decision edges of two candidate graphs are updated as the decision tree grows. 

Identifying valid split moves offers several potential benefits: it ensures the validity of data bipartition, improves the accuracy of tree structure prior and proposal probability calculations (see Supplementary Section \ref{app_subsec:GSB-prior}), and saves memory space for storing marginal likelihoods of all possible moves. Additionally, we shall show that identifying valid decision edges does not increase the computational burden, as it can be simultaneously achieved during the recursive split log marginal likelihood ratio calculations.

\medskip
\noindent\textbf{Computations of split marginal likelihood ratios.} 
At the root node, for each edge of an arborescence candidate graph $\dG=(V, \dE) \in \dbmG$, we determine whether it yields a valid split move and evaluate the split loglikelihood ratio for each valid split. We sample a valid split move following Alg.~\ref{alg:one-weak-learner} and update the cluster membership accordingly. An example is demonstrated in Figure \ref{fig:vertex edge type}($b_1$). However, when decision tree further grows,  each vertex of an arborescence could be associated with data instances belonging to different leaf nodes, and some edges of a candidate arborescence could become invalid with respect to data instances of each leaf node. 
See Figures \ref{fig:vertex edge type} ($b_2$) and ($b_3$) for an example.
Naively reconstructing candidate arborescences for data instances of each leaf node and evaluating possible splits on these new candidate arborescences are computationally cumbersome. To facilitate computation,  we keep the graph structures of candidate arborecences intact for each $\calT$, updating only the cluster memberships of data instances of each vertex and edge types for each splittable leaf node when updating $\calT$.     


\begin{figure}[h]
    \centering
    \includegraphics[width=6in,height=3in]{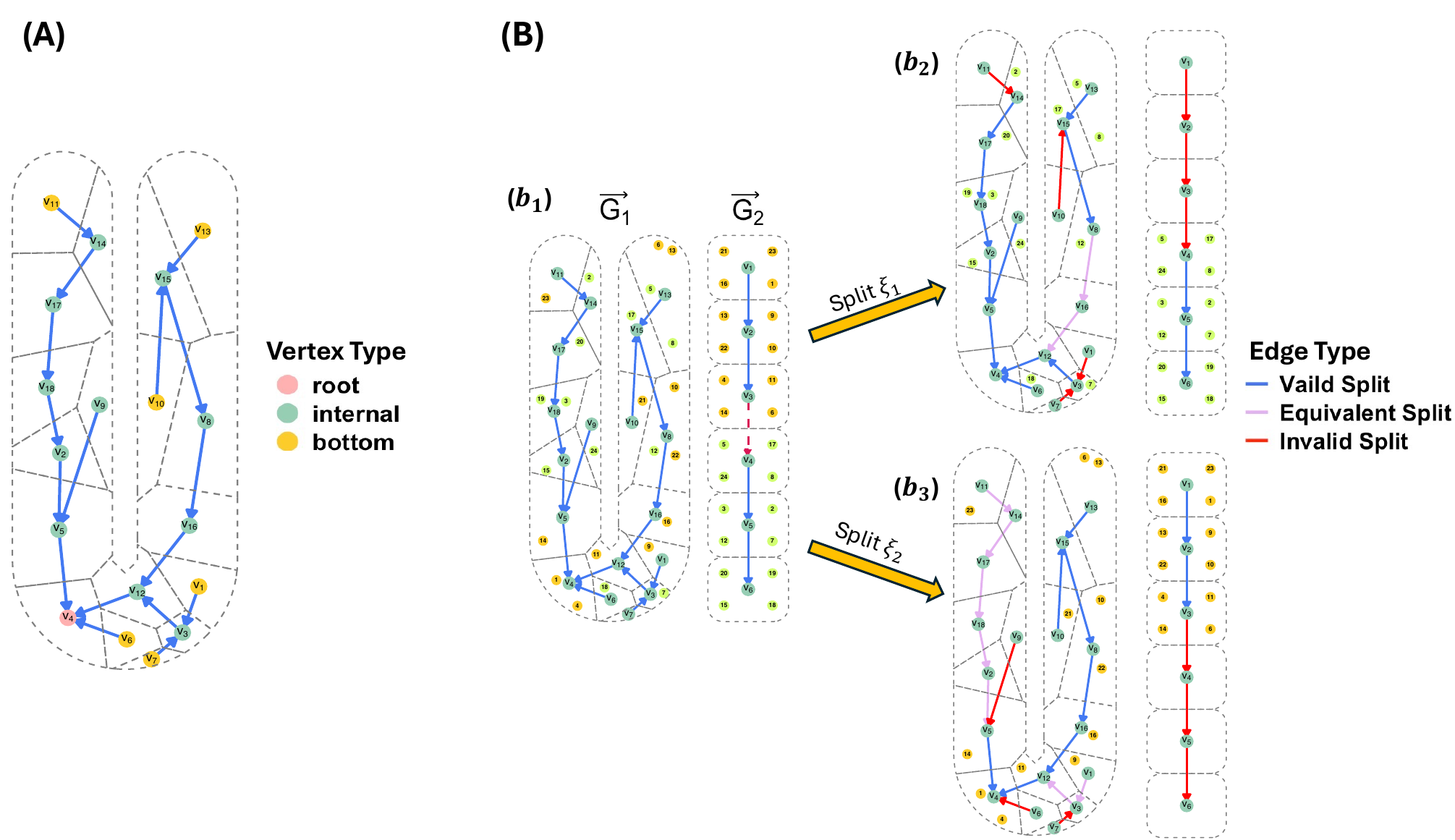}
    \caption{(A): An example of the ordered directed edges. Different color denotes different types of vertices; (B): (b1) shows a decision tree after splitting the root node into two leaf nodes $\xi_1$ (lime dots) and $\xi_2$ (orange dots) by the edge in $\protect\dG_2$ marked by the red dashed line; 
    (b2) and (b3) show the edge types when $\xi_1$ or $\xi_2$ is selected for splitting, respectively.} 
    \setlength{\belowcaptionskip}{-50pt}
    \label{fig:vertex edge type}
\end{figure}

Suppose the current decision tree $\calT$ has $\ell$ leaf nodes, denoted as $\underline{\xi} = (\xi_1, \ldots, \xi_{\ell})$. The training data associated with each vertex $v$ of $\dG$ are grouped according to their leaf node memberships, denoted as  $(v \cap \xi_1 \cap [n_{\text{train}}], \ldots, v \cap \xi_{\ell} \cap [n_{\text{train}}] )$. Suppose we choose $\de_v$ from $\dG$ as a candidate decision edge to split a leaf node $\xi_k$ in $\underline{\xi}$. Disconnecting $\de_v$ from $\dG$ would result in a bipartition of $\dG$ into two sub-arborescences. We refer to the sub-arborescence consisting of $v\cup \operatorname{desc}(v)$ as the right arborescence, and the other as the left arborescence. Let $\xi_{k, R, v}^{*}$ denote the right offspring node after splitting $\xi_k$, and define $\underline{\xi}^{*}_{R, v} = (\xi^{*}_{1, R, v}, \ldots, \xi^{*}_{\ell, R, v})$. 
We define $\underline{\xi}^{*}_{L, v}$
in a similar way. 
The log marginal likelihood ratio for each split move provided in Lemma 1 and \eqref{app_eq:split likelihood ratio} in Appendix involves the computations of $J(\xi_k)$, $H(\xi_k)$, $J(\xi^*_{k,R,v})$, $J(\xi^*_{k,L,v})$, $H(\xi^*_{k,R,v})$, and $H(\xi^*_{k,L,v})$. 
We first compute the gradient statistics $J$ and $H$ at each vertex and leaf node and define two vectors of length $\ell$, denoted as $ \underline{J}(v) = (J^{1}(v), \ldots, J^{\ell}(v)) $ and $\underline{H}(v) = (H^{1}(v), \ldots, H^{\ell}(v))$, respectively, where
\begin{equation}\label{eq:JH-v}
\begin{aligned}
J^{k}(v) = \sum_{i \in v \cap \xi_{k} \cap [n_{\text{train}}] } \Bigl(\dot{\mathcal{L}}_{\sigma}(\hat{\phi}_{i}) - \ddot{\mathcal{L}}_{\sigma}(\hat{\phi}_{i}) \Bigl), 
H^{k}(v) = -\sum_{i \in v \cap \xi_{k} \cap [n_{\text{train}}]} \ddot{\mathcal{L}}_{\sigma}\left(\hat{\phi}_{i}\right),   k = 1, \ldots, \ell
\end{aligned}
\end{equation}

We then compute the vector of leaf-node level gradient statistics, $\underline{J}(\underline{\xi})$ and $\underline{H}(\underline{\xi})$, by 
\begin{equation}\label{eq:JH-overall}
\begin{aligned}
\underline{J}(\underline{\xi}) = \sum_{v \in V} \underline{J}(v), \quad \underline{H}(\underline{\xi}) = \sum_{v \in V}  \underline{H}(v). 
\end{aligned}
\end{equation}

When $\de_v$ is a bottom edge, the right arboresence only contains one vertex $v$, and hence $\underline{J}(\underline{\xi}_{R, v}^{*})=\underline{J}(v)$, $\underline{H}(\underline{\xi}_{R, v}^{*})=\underline{H}(v)$,  $J(\xi_{L, v}^{*})=\underline{J}(\underline{\xi})-\underline{J}(v)$, and $H(\xi_{L, v}^{*})=\underline{H}(\underline{\xi})-\underline{H}(v)$, which can be computed following the expressions in \eqref{eq:JH-v} and \eqref{eq:JH-overall}. 
Beginning from these bottom edges, $\underline{J}$ and $\underline{H}$ associated with the right offspring node $\underline{\xi}^*_{R, v}$ follow a recursive form as:
\begin{equation}\label{eq:JH-root-r}
\begin{aligned}
\underline{J}\left(\underline{\xi}_{R, v}^{*}\right)  = \underline{J}(v) +\sum_{v \in \operatorname{desc}\left(v\right)} \underline{J}(v), \quad
\underline{H}\left(\underline{\xi}_{R, v}^{*}\right)  = \underline{H}(v) + \sum_{v \in \operatorname{desc}\left(v\right)} \underline{H}(v).
\end{aligned}
\end{equation}

The algorithm visits all edges and passes the gradient information from the leaf to the root vertex $v$ recursively to obtain $\underline{\xi}_{R, v}^{*} \cap [n_{\text{train}}]=(\xi_{1,R, v}^{*} \cap [n_{\text{train}}],\ldots, \xi_{\ell, R, v}^{*} \cap [n_{\text{train}}])$, $\underline{J}\left(\underline{\xi}_{R, v}^{*}\right)$ and
$\underline{H}\left(\underline{\xi}_{R, v}^{*}\right)$. 
 $\underline{J}\left(\underline{\xi}_{L, v}^{*}\right)$ and $\underline{H}\left(\underline{\xi}_{L, v}^{*}\right)$ are then easily computed by $\underline{J}\left(\underline{\xi}_{R, v}^{*}\right)=\underline{J}\left(\underline{\xi}\right)-\underline{J}\left(\underline{\xi}_{L, v}^{*}\right)$ and $\underline{H}\left(\underline{\xi}_{R, v}^{*}\right)=\underline{H}\left(\underline{\xi}\right)-\underline{H}\left(\underline{\xi}_{L, v}^{*}\right)$ . 
This algorithm also allows us to determine whether $\de_v$ is a valid decision edge with regard to a leaf node $\xi_k$ by checking whether both $ | \xi_{k, R, v}^{*} \cap [n_{\text{train}}] | $ and $ |\xi_k \cap [n_{\text{train}}]|- |\xi_{k, R, v}^{*}\cap [n_{\text{train}}]|  $ are greater than zero and whether $\de_v$ is redundant. 

Supplementary Alg.S1 presents the detailed algorithm to compute the log marginal likelihood ratios of all possible splits on a single arborescence. 
Note that such computation needs to be repeated for each candidate arborescence in the candidate graphs set $\dbmG$. We can perform a parallel computation over arborescences in $\dbmG$ to further reduce computations. 

After obtaining the log marginal likelihood ratios for every valid split move, we combine them with the log marginal likelihood ratios for every possible merge, as well as the prior ratios and proposal probabilities to calculate the weight function $\eta$ in the IIT algorithm. We provide a detailed example of applying Alg.S1 in Supplementary Figure~\ref{app_fig:alg demo}. 

\begin{remark} 
We only use the selected decision edge to form the bipartition of samples. We do not remove it from the selected arborescence $\dG$, as breaking the connectivity of $\dG$ would cause complications in computing the marginal likelihood ratio and updating cluster memberships for subsequent splits.  We explain the issue in Supplementary Section~\ref{app_subsec:graph-int}. 
\end{remark}

\medskip
\noindent\textbf{Computational complexity.}
In GS-BART, each informed split proposal requires evaluating likelihoods for up to $2\cdot p \cdot (|V|-1)$ candidate splits across the two newly created child nodes after split moves, giving a computational complexity of $\Theta(|V|)$ per proposal if $|V|\gg p$. Merge moves have the same order. Parallelization across $N_{\mathrm{core}}$ cores reduces cost by a factor of $N_{\mathrm{core}}$. By comparison, BART split moves involve choosing a leaf, a variable, and a cut point, then reassigning data, which also yields $\Theta(|V|)$ cost. 

\subsection{Bayesian Function  and Posterior Inference and Prediction}\label{subsec:prediction}

\noindent\textbf{Prediction.}
When constructing a candidate graph, each data point, training or testing, is assigned to a vertex based on the feature information.  During the Bayesian sampling of the GS-BART model parameters, we recursively split both training and testing data into decision tree leaf nodes. We then use the corresponding posterior samples of leaf weights from each weak learner to obtain predictive samples of $\phi_{\text{test}}$ for each testing data and other predictive quantities of interest for the generalized response model. 

\noindent\textbf{Variable importance and partial feature effects.}
We count the frequency with which candidate graphs constructed from each feature are selected to form decision rules. We use this frequency as a metric for evaluating variable importance, as it reflects the dependency between the feature and the response variable. Moreover, although the regression mean function of GS-BART lacks a closed functional form, the posterior samples of the additive trees enable visualization of the partial effects of features of interest to investigate their relationships with the response variable. For a particular set of features of interest, we draw posterior samples of regression mean functions over a range of values for the selected feature(s), while fixing other features at specific values or averaging them out. 

\section{Experiments}\label{sec:experi}

\textbf{Data response models.}  We evaluate the performance of GS-BART on synthetic and real datasets with three response types: normal regressions for continuous responses, multinomial models for the classification of categorical responses, and the variance model for count data responses which assumes $Y_i  \sim \text{Normal}\left\{ \phi_i, \phi_i\right\}$ following \cite{linero2024generalized}. The details of each response model, including the expressions of $\dot{\mathcal{L}}_{\sigma}(\phi)$ and $\ddot{\mathcal{L}}_{\sigma}(\phi)$, the prior hyperparameter settings, and the initial values, are given in Supplementary Section~\ref{app_sec:exp-setting}.

\noindent\textbf{Competing methods.}   
We compare GS-BART with (1) \textbf{XGBoost} using the \textit{caret} package in R; (2) \textbf{BART} using the \textit{BART} package for continuous and classification datasets and \textit{FlexBartVar} package \citep{linero2024generalized} for count datasets; (3) \textbf{LGP} using the latent nearest neighbor Gaussian Process regression in the \textit{BRISC} package~\citep{saha2018brisc} for continuous data and \textit{glmmfields} package~\citep{anderson2019black} for count data; (4) \textbf{RF-GLS} using the spatial random forest method in the \textit{RandomForestsGLS} package~\citep{saha2023random} for continuous datasets; (5)
\textbf{XBART} in the \textit{stochtree} package ~\citep{he2023stochastic} for continuous and classification datasets;
(6) \textbf{flexBART} package \citep{deshpande2024flexbart} for continuous datasets.
(7) \textbf{AddiVortes} in~\cite{stone2024addivortes} for continuous datasets. We report results for RF-GLS and flexBART only on simulation datasets, due to numerical instabilities with Cholesky decompositions and excessive runtimes on real datasets, respectively.  

\noindent\textbf{Simulation datasets.}
We generate synthetic datasets from the model in \eqref{eq:def-likelihood} for each of the three response types and for each of the 2-D U-shape and 3-D bitten Torus domains, respectively.  We assume the true latent regression functions depend on spatially structured features and some unstructured numerical features. To have a fair comparision with XGBoost, BART, and XBART, we also consider a normal response model using the Friedman function (Friedman, 1991) only involving unstructured features on Euclidean domains as the true regression mean function. For the U-shape and Torus examples, we include several irrelevant unstructured numerical features.  The details of the simulation data generations are deferred to Supplementary Section \ref{app_sec:data-intro}.  

\begin{table}[h]
    \begin{center}
    \small
    \begin{threeparttable}
    
    \begin{tabular}{lcccc}
    \toprule[1pt]
    \multicolumn{5}{c}{\textbf{CONTINUOUS:} MSPE $\times 10$ (sd $\times 10$)} \\
    \hline
    \textbf{DATASETS} & \textbf{GS-BART} & \textbf{BART} & \textbf{XGBOOST} & \textbf{LGP} \\
    \hline
    U-SHAPE ($n = 800$, $p = 5$) & \textbf{0.82} (0.22) & 3.51 (0.83) & 8.17 (1.14) &  3.61 (0.15)      \\
    TORUS  ($n = 800$, $p = 5$)  & \textbf{10.20} (2.78) & 21.37 (3.70) & 16.77 (1.48) &  24.80 (0.24)  \\
    FRIEDMAN ($n = 400$, $p = 5$)  & \textbf{13.31} (1.95) & 16.27 (2.41) & - &  -  \\
    \hline 
     &  \textbf{RF-GLS} & \textbf{XBART} & \textbf{flexBART} & \textbf{AddiVortes}   \\
    \hline
    U-SHAPE ($n = 800$, $p = 5$)  & 2.87 (0.15) & 3.75 (0.61)  & 1.29 (0.26) & 2.14 (0.62)\\
    TORUS ($n = 800$, $p = 5$)  & 19.12 (0.91)& 18.50 (1.45)  & 16.22 (3.32) & 12.15 (2.22)  \\
    FRIEDMAN ($n = 400$, $p = 5$)  & - & 28.14 (3.29)  & 17.11 (2.31) & 27.71 (3.73) \\
    \hline
    \multicolumn{5}{c}{\textbf{COUNT:} RMSPE (sd)} \\
    \hline
    \textbf{DATASETS} & \textbf{GS-BART} & \textbf{BART} & \textbf{XGBOOST} & \textbf{LGP}   \\
    \hline
    U-SHAPE ($n = 800$, $p = 5$) &  \textbf{2.23} (0.25) & 3.48 (0.56) & 3.61 (0.48)  & 4.86 (0.43)   \\
    TORUS ($n = 800$, $p = 5$) & \textbf{1.94} (0.15) & 1.98 (0.17) & 2.72 (0.23) & 3.61 (0.24)  \\
    \hline
    \multicolumn{5}{c}{\textbf{CLASSIFICATION:} ACCURACY \% (sd $\times  10$)} \\
    \hline
    \textbf{DATASETS} & \textbf{GS-BART} & \textbf{BART} & \textbf{XGBOOST}  & \textbf{XBART} \\
    \hline
    U-SHAPE ($n = 800$, $p = 5$) & 91.47 (0.16) & 87.40 (0.15) & \textbf{91.62} (0.17) &  89.58 (0.16) \\
    TORUS ($n = 800$, $p = 5$) &  \textbf{80.68} (0.18) & 78.45 (0.15)  & 80.47 (0.16) & 80.33 (0.18) \\
    \bottomrule[1pt]
    \end{tabular}
    \setlength{\belowcaptionskip}{-25pt}
    \caption{Average out-of-sample prediction performance metrics with standard errors (in parentheses) across 50 repeated experiments for each synthetic data example. 
    `-' denotes N$/$A.} \label{Tab:SimulationDataRes}
    
    \end{threeparttable}
    \end{center}
\end{table}

\noindent\textbf{Real datasets.} 
We consider real datasets with different response and feature types. 
We use two datasets from the GeoDa Data and Lab \citep{anselin2009geoda}. The first is the New York City income \& education $(n = 1690, p=5)$. We take the logarithm of income as the response and fit the model using latitude and longitude as spatially structured features and $p$ social demographic covariates including education level as the unstructured features. 
The second data is the home sales data $(n = 20149, p=18)$ in King County, Seattle, Washington, where we treat the logarithm of the house sales as response and $p$ covariates as the unstructured features in addition to spatial features. 
For network feature data, we consider the Cora citation network data with $n =  2780$ nodes, each representing a scientific publication, and 5429 edges, representing citation relationships between papers. We also analyze four other county-level real datasets, treating the county adjacency graph as a network feature. 
We model the logit-transformed proportion of county-level Republican voters in the 2016 United States presidential elections $(|V|=3071, |E|=8669, p=55)$ by a normal response model. Additionally, we model the county-level count of all cancer types $(|V|=3103, |E|=8728, p=27)$ by the variance model for count response using population divided by 200,000 as an offset. For classification problems, we consider two county-level binary responses, the presence/absence of unhealthy air quality days in 2023 and flood in 2016, with $(|V| = 893, |E| = 8754, p= 27)$ and $(|V|=3037, |E|=11601, p=27)$, respectively. The difference in $|V|$ in these county-level data is due to missing values in some response variables. More details are provided in Supplementary Section \ref{app_sec:data-intro}.

\begin{table}[h]
    \small
    \begin{center}
    \begin{threeparttable}
    \begin{tabular}{lcccccc}
    \toprule[1pt]
    DATASETS & \textbf{GS-BART} & \textbf{XGBOOST}  & \textbf{LGP}  & \textbf{XBART} & \textbf{AddiVortes} & \textbf{BART} \\
        \hline & \multicolumn{5}{c}{ \textbf{CONTINUOUS}, MSPE/MSPE(BART)} & \\
        \hline NYC Education & $\mathbf{0.85}$ & 0.96 & 1.24 & 1.54 & 0.96 & \fbox{0.067}  \\
        US Election        & $\mathbf{0.76}$ & 0.98 & 2.08 & 1.22 & 0.88 & \fbox{0.075} \\
        \multirow{2}{*}{\shortstack[l]{KC House \\ (Runtime)}}  
        & $\mathbf{0.34}$ & 0.99 & 1.00  & 1.00 & 1.01 & \fbox{0.276}\\ 
        &  (841)  & (1339)  & (78) & (19) & (3807)& (475)\\ 
        \hline & \multicolumn{5}{c}{ \textbf{COUNT}, RMSPE/RMPE(BART)} & \\
        \hline  \multirow{2}{*}{\shortstack[l]{US Cancer \\ (Runtime)}} 
        & $\mathbf{0.71}$ & 0.82 & 1.27 & - & -  & \fbox{32.30} \\
        & (190)  & (103) & (4$\times 10^4)$ & - & - & (201)\\ 
        \hline & \multicolumn{5}{c}{ \textbf{CLASSIFICATION}, ACCURACY (\%)}& \\
        \hline Air Pollution & 81.81  & 78.97 & 80.68 & 79.55 & -     & $\mathbf{82.95}$  \\ 
        Flood  & 61.94 & 62.11 & 59.31 & $\mathbf{63.76}$ & -     & 58.98 \\ 
        Cora Network & $\mathbf{77.98}$ &  75.76 & - & 73.74 & -     & 70.71 \\
    \bottomrule[1pt]
    \end{tabular}\setlength{\belowcaptionskip}{-20pt}
    \caption{Out-of-sample prediction performance on real datasets. For continuous and count outcomes, metrics are reported as ratios relative to BART (with the BART column showing the original values). Runtimes (seconds) are reported for the King County House ($n=20419,p=18$) and US Cancer Count ($n=3103,p=27$) data.}
    \label{Tab:Real data res ratio}
    \end{threeparttable}
    \end{center}
\end{table}

\noindent\textbf{Experiment settings.}  
We use the method introduced in Example 3 of Section~\ref{subsec:candidategraph} to construct candidate graphs for the U-shape, Torus, and King County House datasets, and the method in Example 2 for the other datasets. For each data example, unless otherwise specified, we set the number of trees to 50 across all decision tree methods. GS-BART and XBART are run for 215 iterations with a 15-iteration burn-in, while BART is run for 6000 iterations with a 2000-iteration burn-in. Supplementary Section~\ref{app_subsec:GSB-mixing} shows that IIT substantially improves the mixing of GS-BART compared to the standard random-walk split-merge sampler in BART, resulting in a much higher effective sample size and requiring fewer iterations and burn-in. To evaluate and compare the out-of-sample prediction performance of the methods, we use an 80/20 training and testing split to evaluate and compare the out-of-sample prediction performance of the methods. We adopt mean squared error, root mean squared error, and prediction accuracy as evaluation metrics for continuous, count, and classification datasets, respectively.

\noindent\textbf{Results.}
Tables \ref{Tab:SimulationDataRes} and \ref{Tab:Real data res ratio} show the out-of-sample prediction results of different methods on synthetic and real datasets, respectively. Overall, GS-BART outperforms other competing methods on most synthetic and real datasets, especially those with continuous and count response data. In classification datasets, the prediction accuracy of GS-BART is slightly better than or comparable to other methods. In Friedman's example which only involves unstructured features, the decision tree prior model of GS-BART is equivalent to BART and XBART, but we still observe an improved performance of GS-BART, potentially due to the better mixing of IIT. 
These results demonstrate the consistent superiority and generalizability of GS-BART.
In contrast, many other competing methods are only applicable to some particular response models or exhibit numerical instability. Table \ref{Tab:Real data res ratio} also reports the computational runtimes of each method on two large real datasets. All experiments are conducted on an Apple M1 Pro chip using 9 cores for GS-BART. The computation of GS-BART is modestly slower than that of standard BART for normal regressions and is slightly faster than the generalized BART~\citep{linero2024generalized} for count regressions, and its overall computational cost is affordable for moderate to large-scale real datasets.

It is noticeable that GS-BART is particularly suitable for analyzing response data that are highly dependent on graph structures or lie on complex domains, as shown in the U-shape, Torus, and King County House examples. As an example, we examine the prediction surfaces of the U-shape data in Figure \ref{fig:Ushape results}(A). The U-shape response data is generated using a true piecewise smooth function where the discontinuity circular boundary separates the U-shape into three subregions, including two arm areas and one circle. 
Axis-parallel artifacts are clearly observed around the true discontinuity circular boundary in BART. 
Although such artifacts do not appear in LGP and RF-GLS, it is obvious that these two Gaussian process based approaches produce overly smoothed predictions near the circular boundary and fail to capture the abrupt changes of the true function.  In contrast, the GS-BART method successfully captures the nonlinear discontinuity boundaries and varying smoothness of the true function. Moreover, the other competing methods give poorer predictions than GS-BART at some locations near the U-shape domain boundary between the two arms. For BART, the issue is caused by their axis-parallel split rule, which forces some boundary points in different arms to be in the same cluster. For LGP and RF-GLS, the problem is partly due to the use of Euclidean distance-based covariance functions, which impose undesirable smoothness between boundary points that are close in Euclidean space but actually belong to two different arms. 

\begin{figure}[h]
    \centering
    \includegraphics[scale=.62]{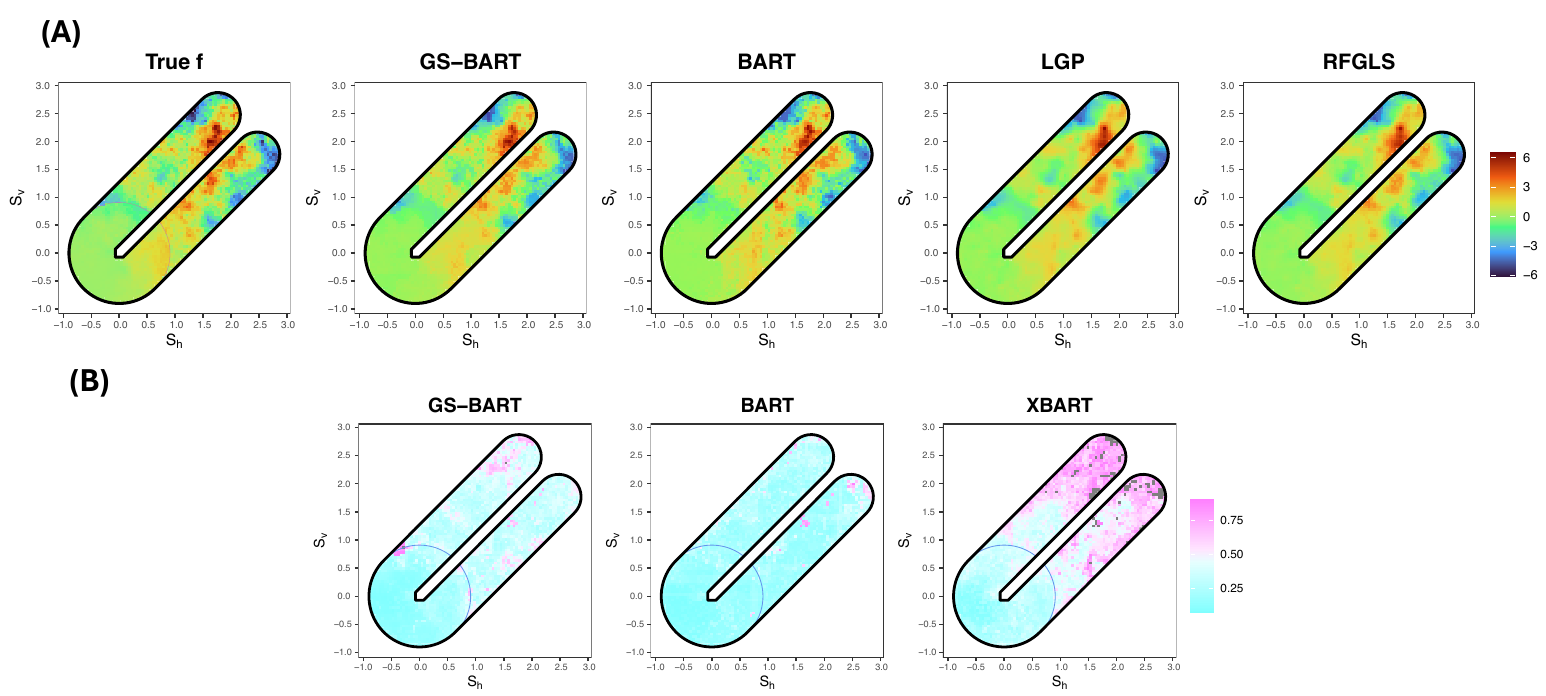}
    \setlength{\belowcaptionskip}{-10pt}
    \setlength{\abovecaptionskip}{-10pt}
    \caption{(A) True and predictive spatial surface plots of $\phi(\textbf{s}, \textbf{x})$ for each method using one simulated normal response U-shape dataset with $n=800$, $p=2$ and $\sigma=0.15$. 
    The discontinuity in the true function is marked by circular lines. 
    (B) Posterior predictive standard deviation surface plots of GS-BART, BART and XBART for the same U-shape example. Standard deviations larger than 0.9 are marked by gray. 
    }
    \label{fig:Ushape results}
\end{figure}

To further investigate the performance of posterior uncertainty quantification of GS-BART and other competing Bayesian additive regression tree-based approaches, we present the predictive standard deviation surface maps in Figure \ref{fig:Ushape results}(B). We also find the 95\% highest posterior density interval (HDI) at each location and compute the coverage rate and average interval length.  
Among the three methods, BART seems to underestimate the posterior uncertainty, potentially due to the poor mixing of the standard random-walk split-merge-type sampler; it has the shortest HDI length of 0.60 but only achieves a relatively low coverage rate of  $58.2\%$. In contrast, XBART shows high posterior predictive standard deviation in the two arm regions of the U-shape, with an average HDI length of 1.78 and a coverage rate of $93.6\%$. GS-BART, however, provides uncertainty estimates that match well with the truth, showing high posterior predictive standard deviation near discontinuities and in regions where the true function fluctuates significantly. It achieves a $90.2\%$ coverage rate with a much shorter HDI length of 0.97 compared with XBART.

Similar to other additive decision tree methods, GS-BART can be utilized to evaluate feature importance and interaction among features as mentioned in Section~\ref{subsec:prediction}. For instance, in the U-shape example where the true function only depends on $\bfs$ and $x_1$, spatial features and $x_1$ exhibit the highest and second highest variable importance estimated from GS-BART, respectively; 84\% of the splits have their decision rules determined by the candidate graphs constructed from $\bfs$ and $x_1$, indicating that a significant portion of the true function's variation can be attributed to $\bfs$ and $x_1$.  

We examine the performance of GS-BART on the U-shape example while varying the number of weak learners, $T$, the number of arborescences generated by structured feature, $M$, and the number of vertices $|V|$ for each arborescence. 
Overall, GS-BART showed fairly robust performance to these tuning parameters as long as they are sufficiently large.   
We notice that increasing $T$ and $M$ generally improves the prediction performance for a fixed $|V|$. However, beyond a certain threshold, there appear to be diminishing or negative returns in prediction accuracy. Increasing $|V|$ generally improves prediction accuracy but at the expense of having a higher computational cost. Further increasing it does not result in substantial additional improvements once it is relatively large.
More details about the sensitivity analysis are described in Supplementary \ref{app_subsec:Ushape-add}. 

To illustrate the utility of GS-BART for investigating the functional relationship between features and response data, we compare the partial dependence plots showing the spatial feature effect on the King County log house sales price under BART and GS-BART in Figure \ref{fig:spatial partial dependence plot}.   
Axis-parallel discontinuities can be clearly observed in the estimated spatial partial effects from BART. In contrast, the result of GS-BART exhibits spatial patterns with more flexible discontinuity boundaries and meaningful interpretations. For instance, the observed house sales price decline around the Washington 203 highway from both methods. However, this highway passes through only the top right part of our study area in a vertical direction but the estimated function discontinuity of BART extends vertically to the entire area and causes some artifacts in the upper region without the highway. 
Additionally, houses around lakes usually have higher sale prices. GS-BART successfully detects such a neighborhood around Lake Sawyer where the spatial effect is higher than the surrounding areas, a pattern overlooked by BART. 
Besides, the spatial effects around the University District are relatively high for both methods. However, the result of GS-BART exhibits more detailed spatial distribution patterns that BART does not capture. Specifically, GS-BART captures the pattern that the  spatial effect in the areas near Interstate 5 tends to decrease in the University District, while the result of BART shows little variation in this region. 
We also show the partial dependence plot examining the effect of living area on log house sales price in Supplementary Section~\ref{app_subsec:KingHouse-add}.

\begin{figure}[ht]
    \centering
    \includegraphics[scale=.6]{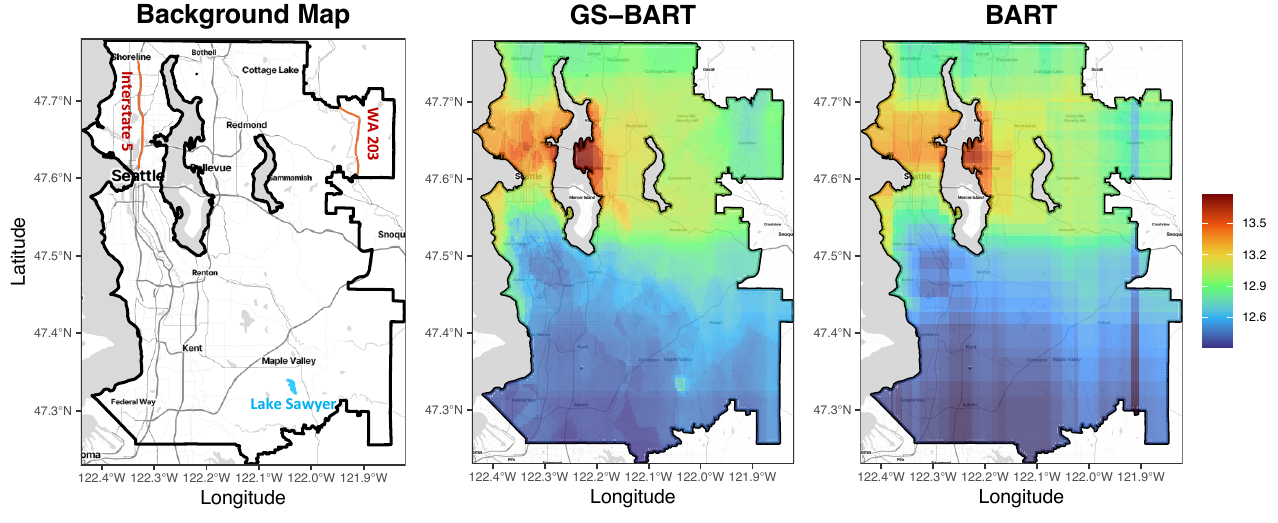}
    \setlength{\belowcaptionskip}{-10pt}
    \caption{The partial dependence plots showing spatial effects on log house sales price for GS-BART and BART.}
    \label{fig:spatial partial dependence plot}
\end{figure}

\section{Conclusion}\label{sec:conclusion}
We propose an extension of BART to model generalized response data with graph-structured features.  We encode feature information by candidate graphs and incorporate them in additive tree models through a graph split decision rule. The method is particularly well-suited for analyzing spatial and network data, especially those with complex geometries, as evidenced by its superior performance in our extensive numerical examples. It offers several advantages over existing spatial and network regression models, including the ability to capture both varying smoothness and abrupt changes in the regression function, accommodate high-dimensional covariates, handle various response types, and capture nonlinear partial effects of features and their interactions. 
 
Moving forward, there are several avenues for future research. Currently, our research assumes that the manifold structure or graph relations among data are predefined.  One promising area of research is to explore methods for learning unknown manifolds or graphs directly from the observed high-dimensional covariates~\citep[see, e.g.,][]{arvanitidis2019fast,niu2023intrinsic}. This would greatly expand the applicability of GS-BART in regression problems where covariates have intrinsic but unknown multivariate or graph relations.    
GS-BART offers a nonparametric method to model latent functions involving graph-structured inputs, which is highly versatile and can potentially be embedded into various hierarchical Bayesian models beyond those considered in this paper, including point pattern data~\citep{dahl2024modeling} and varying coefficient models. 
Moreover, there is potential for further enhancing the computational efficiency of GS-BART to handle larger-scale problems beyond those considered in this paper, such as using a subsampling-based IIT method~\citep{li2023importance} that constructs the informed proposal only using a subset of training data. 
Column subsampling is another popular strategy used in XGBoost and Random Forest to handle computations with high-dimensional features, where each decision tree is trained only on a random subset of features. It is convenient to implement similar column subsampling strategies in GS-BART. 
For example, we may choose a different subsample of candidate arborescences to construct each $\dbmG_t$ to reduce computation. We provide a preliminary study on the performance of column (i.e., candidate graph) subsampling in Supplementary Section~\ref{app_sec:col-sub}. Finally, on the theoretical side, it will be interesting to provide the Bayesian posterior concentration theories following~\cite{rovckova2020posterior} to justify the superior empirical performance of GS-BART in function estimations and offer insights for further improvement of the methods.   
\vspace{-2mm}
\begin{description}  
\item[Supplementary File:] The supplementary consists of  
(1) Additional details and proofs of the GS-BART algorithms; (2) Detailed data generation and model settings for numerical experiments; (3) Detailed descriptions of synthetic and real datasets; (4) Additional results for U-shape and King County House sales example. (5) Additional results for mixing behavior and column subsampling. 
\end{description}

\appendix
\section*{Supplementary Material}
\setcounter{section}{0}
\renewcommand{\theequation}{S\arabic{equation}}
\renewcommand{\thetable}{S\arabic{table}}
\renewcommand{\thefigure}{S\arabic{figure}}
\renewcommand{\thealgorithm}{S\arabic{algorithm}}
\renewcommand\thesection{S\arabic{section}}
\renewcommand\thesubsection{S\arabic{section}.\arabic{subsection}}
\newcommand{\appropto}{\mathrel{\vcenter{
			\offinterlineskip\halign{\hfil$##$\cr
				\propto\cr\noalign{\kern2pt}\sim\cr\noalign{\kern-2pt}}}}}

\section{GS-BART Algorithm}
\subsection{Validity of the IIT Sampler}\label{app_subsec:GSB_val}

At each backfitting step, our goal is to sample a decision tree $\calT_t$ and its associated leaf parameters $\calM_t$ from the posterior $p(\calT_t, \calM_t \mid \vartheta_{-t})$, where $\vartheta_{-t}$ denotes all other model parameters. 
Note that we omit conditioning on data in the posterior distributions below when there is no risk of confusion.  

\paragraph{Gaussian case.}  
For the Gaussian case, the informed proposal  is proportional to 
\begin{equation*}
	\eta(\calT^{*}_t | \calT_t, \vartheta_{-t})=q(\calT^{*}_t | \calT_t) h\left(\frac{ p(\calT^*_t  | \vartheta_{-t}) q(\calT_t | \calT^{*}_t)}{ p(\calT_t   | \vartheta_{-t})  q(\calT^{*}_t | \calT_t)}\right), 
\end{equation*}
where $q(\cdot)$ is an uninformed proposal kernel (see Section~\ref{app_subsec:GSB_val}) and $h(\cdot)$ is a balancing function satisfying $h(x) = x h(1/x)$.  
It results in a Markov chain with transition matrix  $\bm{P} (\calT_t , \calT^{*}_t)=\eta(\calT^{*}_t | \calT_t, \vartheta_{-t} ) / Z(\calT_t, \vartheta_{-t})$ for each $\calT^{*}_t \in \mathcal{N}(\calT_t)$.
We can easily show that  
\begin{eqnarray*}
	p(\calT_t | \vartheta_{-t})\, \eta(\calT^{*}_t | \calT_t, \vartheta_{-t} ) &=& 
	p(\calT_t | \vartheta_{-t})q(\calT^{*}_t | \calT_t) h\left(\frac{ p(\calT^*_t  | \vartheta_{-t}) q(\calT_t | \calT^{*}_t)}{ p(\calT_t   | \vartheta_{-t})  q(\calT^{*}_t | \calT_t)}\right)\\
	&=& p(\calT_t | \vartheta_{-t})q(\calT^{*}_t | \calT_t) \frac{ p(\calT^*_t  | \vartheta_{-t}) q(\calT_t | \calT^{*}_t)}{ p(\calT_t   | \vartheta_{-t})  q(\calT^{*}_t | \calT_t)} h\left(\frac{ p(\calT_t   | \vartheta_{-t})  q(\calT^{*}_t | \calT_t)}{ p(\calT^*_t  | \vartheta_{-t}) q(\calT_t | \calT^{*}_t)}\right)  \\ 
	&=& p(\calT_t^* | \vartheta_{-t})\, \eta(\calT_t | \calT^{*}_t, \vartheta_{-t} ). 
\end{eqnarray*}
Note that 
$p(\calT_t | \vartheta_{-t})\, \eta(\calT^{*}_t | \calT_t, \vartheta_{-t} )=p(\calT_t   | \vartheta_{-t}) Z(\calT_t, \vartheta_{-t})\bm{P}(\calT, \calT^*)$.
Therefore, $\bm{P}$ is reversible with respect to the distribution  $\tilde{p}(\calT_t   | \vartheta_{-t}) \propto  p(\calT_t   | \vartheta_{-t}) Z(\calT_t, \vartheta_{-t})$, which satisfies the detailed balance condition 
\begin{equation*}
	\begin{aligned}
		\tilde{p}(\calT_t | \vartheta_{-t})\, \bm{P}(\calT, \calT^*)\ = 
		\tilde{p}(\calT_t^* | \vartheta_{-t})\, \bm{P}(\calT^*, \calT). 
	\end{aligned}
\end{equation*}
To correct for the bias caused by the term $Z(\calT_t, \vartheta_{-t})$, after collecting samples from the informed proposal, we draw a sample from $p(\calT_t   | \vartheta_{-t})$ using importance reweighted sampling with importance weight $1/Z(\calT_t, \vartheta_{-t})$. Conditional on a sample of $\calT_t$, we draw $\calM_t$ from its exact conditional normal posterior distribution.   

\paragraph{Non-Gaussian case.} 
We use the second-order approximation for the log likelihood function in the generalized regression problems (see Lemma 1). We note that this approximation is only applied to the informed proposal distribution instead of the target distribution. 
We now show that by accounting for this approximation in the importance weight calculation, we can still ensure that this sampling scheme targets the exact posterior distribution rather than the approximate one.

Let $\hat{p}(\calT_t, \calM_t \mid \vartheta_{-t})=  \hat{p}(\calT_t | \vartheta_{-t})\, \hat{p}(\calM_t | \calT_t, \vartheta_{-t}) \propto \prod_{i \in [n_{\text{train}}]}\hat{L}_\sigma(\phi_i)\,
\pi(\calT_t)\, \pi(\calM_t | \calT_t,\sigma_\mu)$ be the approximate joint posterior distribution of $(\calT_t,\calM_t)$ under the second-order approximation of the log likelihood, where 
$\hat{p}(\calT_t  | \vartheta_{-t})$  and $\hat{p}(\calM_t \mid \calT_t, \vartheta_{-t})$ are the approximated conditional posterior of $\calT_t$ and $\calM_t$ computed following Lemma 1, respectively, and 
$\hat{L}_\sigma(\phi_i)=\exp \Bigl\{\mathcal{L}_{\sigma}(\hat{\phi}_{i}) +   ( \mu_{t,k} - \hat{\phi}_{t,i} ) \dot{\mathcal{L}}_{\sigma}(\hat{\phi}_{i})  +   (\mu_{t,k} - \hat{\phi}_{t,i})^2  \ddot{\mathcal{L}}_{\sigma}(\hat{\phi}_{i}) /2 \Bigl\}$. 
We first draw $\calT^*_t$ from the informed proposal 
\begin{equation}\label{eq:def-weight-nongaussian}
	\eta(\calT^{*}_t | \calT_t, \vartheta_{-t})=q(\calT^{*}_t | \calT_t) h\left(\frac{ \hat{p}(\calT^*_t  | \vartheta_{-t}) q(\calT_t | \calT^{*}_t)}{ \hat{p}(\calT_t   | \vartheta_{-t})  q(\calT^{*}_t | \calT_t)}\right),
\end{equation}
and then draw $\calM^*_t$ from the approximated normal posterior $\hat{p}(\calM_t^* \mid \calT_t^*,\vartheta_{-t})$. 

This informed proposal of $(\calT_t,\calM_t)$ results in a Markov chain with transition matrix $\bm{P}\big((\calT_t,\calM_t ), (\calT^*_t,\calM^*_t)\big)=\bm{P}(\calT_t,\calT^*_t)\hat{p}(\calM^{*}_t|\calT^*_t, \vartheta_{-t})$, which is reversible with respect to the distribution 
\begin{equation}\label{eq:def-tilde-p}
	\tilde{p}(\calT_t,\calM_t| \vartheta_{-t}) \propto  \hat{p}(\calT_t,\calM_t   | \vartheta_{-t}) Z(\calT_t,\vartheta_{-t}).
\end{equation}
The proof of the detailed balance condition is a simple modification of the Gaussian case as follows: 

\begin{equation*}
	\begin{aligned}
		& & \hat{p}(\calT_t | \vartheta_{-t})\, \hat{p}(\calM_t | \calT_t, \vartheta_{-t})Z(\calT_t,\vartheta_{-t})\, \bm{P}(\calT, \calT^*)\, \hat{p}(\calM_t^* | \calT_t^*, \vartheta_{-t})  \\
		& = &
		\hat{p}(\calT_t^* | \vartheta_{-t})\, \hat{p}(\calM_t^* | \calT_t^*, \vartheta_{-t})Z(\calT^*_t,\vartheta_{-t})\, \bm{P}(\calT^*, \calT)\, \hat{p}(\calM_t | \calT_t, \vartheta_{-t}).
	\end{aligned}
\end{equation*}

After collecting samples of $(\calT_t,\calM_t)$ from the informed proposal, we draw a sample from $p(\calT_t, \calM_t   | \vartheta_{-t})$ with the importance weight: 
\begin{equation*}
	w(\calT_t,\calM_t) =  \frac{ p(\calT_t,\calM_t   | \vartheta_{-t})  }{ \hat{p}(\calT_t,\calM_t   | \vartheta_{-t})  Z(\calT_t,\vartheta_{-t})}. 
\end{equation*}
Such a scheme targets the exact posterior $p(\calT_t, \calM_t \mid \vartheta_{-t}, \mathbf{y})$, and its validity is irrelevant to the construction of $\hat{p}(\calT_t | \vartheta_{-t})$ and $\hat{p}(\calM_t | \calT_t, \vartheta_{-t}) $ (i.e., it applies to any other approximation scheme as well). 

Since $ p(\calT_t, \calM_t \mid \vartheta_{-t})   \propto \prod_{i \in [n_{\text{train}}]} L_\sigma(\phi_i)\, \pi(\calT_t)\, \pi(\calM_t | \calT_t,\sigma_\mu)$, for our approximation scheme this weight further simplifies to \begin{equation*}
	w(\calT_t,\calM_t) =  \frac{  1}{  Z(\calT_t,\vartheta_{-t})} \prod_{i \in [n_{\text{train}}]} \frac{ L_\sigma(\phi_i)}{ \hat{L}_\sigma(\phi_i)}.  
\end{equation*}

We remark that we only need to compute  
$\prod_{i \in [n_{\text{train}}]} L_\sigma(\phi_i)$  
for the samples drawn from the informed proposal instead of all possible moves, and thus the associated computational cost is negligible compared to the informed proposal scheme.
Note that $L_\sigma$ also has a closed-form expression. 

We also remark that there is a simple strategy for making use of all samples collected by the Markov chain from the informed proposals. By the law of total expectation, for any function $f$ of interest, we have $$\mathbb{E}[ f(\mathbf{X}) \mid \mathbf{y}] = \mathbb{E}[ \mathbb{E}[ f(\mathbf{X}) \mid  \vartheta_{-t}, \mathbf{y} ] \mid \mathbf{y}], $$
where the inner expectation $\mathbb{E}[ f(\mathbf{X}) \mid  \vartheta_{-t}, \mathbf{y} ]$ can be estimated using all (or part of) samples collected by the Markov chain targeting $\tilde{p}(\calT_t,\calM_t   | \vartheta_{-t})$ in~\eqref{eq:def-tilde-p} with importance correction. 

\subsection{Uninformed Proposal Probabilities and Prior Ratios for the Sampling of Trees}\label{app_subsec:GSB-prior}
Below, we give expressions of the uninformed proposal $q(\cdot)$ and the prior ratio involved in the informed proposal in \eqref{eq:def-weight-nongaussian}. 

\noindent\textbf{Split move probability of the uninformed proposal.} 
Given the current tree $\calT$, an uninformed proposal selects a split or merge move with probabilities $\mathbb{P}(\text { SPLIT })$ and $\mathbb{P}(\text { MERGE })$, respectively. 
In the split move, proposing $\mathcal{T}^*$ from the ``neighborhood" of the original tree $\mathcal{T}$ involves selecting a splittable leaf node, a valid decision arborescence, and a valid decision edge to split a current terminal node into two offspring nodes. The split proposal probability takes the form:  
\begin{equation}\label{app_eq:Split prior ratio}
	\begin{aligned}
		q\left(\mathcal{T}^* | \mathcal{T}\right)= & \mathbb{P}(\text { SPLIT }) \mathbb{P}(\text { selecting } \xi \text { to split }) \times \\
		& \mathbb{P}(\text{ selecting  } \dG \text { to form the split} \mid \xi \text{ is selected}) \times \\
		& \mathbb{P}(\text {selecting} \de \text {to form the split} \mid \xi \text{ and } \dG \text{ are selected}) \\
		= & \mathbb{P}(\text { SPLIT }) \frac{1}{\ell_s} 
		\pi_g\left(\dG | \xi \right) \pi_e\left(\de | \dG, \xi\right).
	\end{aligned}
\end{equation}
where $\ell_s$ denotes the number of splittable leaf nodes. 
We set the proposal probabilities for selecting $\dG$, $\xi$, and $\de$ to be the same as those used in the tree split decision rule prior model described in Section~\ref{subsec:modelformu}. For example, a common choice for $\pi_g\left(\dG | \xi \right)$ is a uniform distribution over the subset of $\dbmG$ whose elements have at least one valid edge with regard to $\xi$, and $\pi_e\left(\de | \dG, \xi\right)$ is also a uniform distribution over the subset of $\dE$ containing valid edges with regard to $\xi$.



\noindent\textbf{Merge move proposal probability of the uninformed proposal.} In the merge move, proposing $\mathcal{T}^*$ from the original tree $\mathcal{T}$ involves selecting an internal node $\eta$ of $\calT$ that has two leaf child nodes (called the second generation internal node) and merging its two offspring nodes, whose proposal probability is
\begin{equation}\label{app_eq:Merge prior ratio}
	q\left(\mathcal{T}^* | \mathcal{T}\right) =\mathbb{P}(\text { MERGE }) \mathbb{P}(\text { selecting } \eta \text { to merge its two children})=\mathbb{P}(\text { MERGE }) \frac{1}{w_2^*}
\end{equation}
where $w_2^*$ denotes the number of second-generation internal nodes in $\calT^*$.  

\noindent\textbf{Decision tree prior ratio.}
In Section \ref{subsec:modelformu}, we describe the generative prior for each decision tree. Based on that, the prior ratio between $\calT^*$ and $\calT$ after a split move has the form:
\begin{equation}\label{eq:structure prior ratio}
	\begin{aligned}
		\frac{\pi \left(\mathcal{T}_*\right)}{\pi(\mathcal{T})} & =\frac{\left(1-p_{\text {split}} \left(\xi^*_L\right)\right)\left(1-p_{\text {split}} \left(\xi^*_R\right)\right) p_{\text {split}} (\xi ) p_{\text {rule}}(\xi)}{\left(1-p_{\text {split}}(\xi)\right)} \\
		& =\frac{\left(1-\frac{\alpha}{\left(1+d_{\xi^*_L}\right)^\beta}\right)\left(1-\frac{\alpha}{\left(1+d_{\xi^*_R}\right)^\beta}\right) \frac{\alpha}{\left(1+d_\xi\right)^\beta} \pi_g\left(\dG | \xi \right) \pi_e\left(\de |\dG, \xi \right)}{1-\frac{\alpha}{\left(1+d_\xi\right)^\beta}} \\
		& =\alpha \frac{\left(1-\frac{\alpha}{\left(2+d_\xi\right)^\beta}\right)^2 \pi_g\left(\dG | \xi \right) \pi_e\left(\de |\dG, \xi\right) }{\left(\left(1+d_\xi\right)^\beta-\alpha\right) } .
	\end{aligned}
\end{equation}
The prior ratio between $\calT^*$ and $\calT$ after a merge move can be derived similarly. 

\subsection{Marginal Likelihood Approximation and Parameter Posterior Distributions}\label{app_subsec:GSB-approxi}

In GS-BART, we apply a second-order Taylor expansion to the marginal likelihood at each leaf node such that the approximated marginal loglikelihood has a closed quadratic form. The derivation is as follows: 

\begin{equation*}
	\begin{aligned}
		m_{\xi_k} &= \int \operatorname{N}\left(\mu_{t, k} | \mu_0, \sigma_\mu^2\right) \exp \Bigl\{ \sum_{i \in \xi_k \cap [n_{\text{train}}] } \mathcal{L}_{\sigma}(\phi_i-\hat{\phi}_i + \hat{\phi}_i)\Bigl\} d \mu_{t, k}\\
		\hat{m}_{\xi_k} & = \int \operatorname{N}\left(\mu_{t, k} | \mu_0, \sigma_\mu^2\right) \exp\Bigl\{ \sum_{i  \in \xi_k \cap [n_{\text{train}}] } \Bigl(\mathcal{L}_{\sigma}(\hat{\phi}_{i}) +   ( \mu_{t, k} - \hat{\phi}_{t,i} ) \dot{\mathcal{L}}_{\sigma}(\hat{\phi}_{i})  +   (\mu_{t, k} - \hat{\phi}_{t,i})^2  \frac{\ddot{\mathcal{L}}_{\sigma}(\hat{\phi}_{i})}{2} \Bigl) \Bigl\} d \mu_{t, k} \\ 
		&= \exp\Bigl\{ \sum_{i \in \xi_k \cap [n_{\text{train}}] } \Bigl( \mathcal{L}_{\sigma}(\hat{\phi}_{i}) - \hat{\phi}_{t,i} \dot{\mathcal{L}}_{\sigma}(\hat{\phi}_{i}) +  \hat{\phi}^2_{t,i} \frac{\ddot{\mathcal{L}}_{\sigma}(\hat{\phi}_{i})}{2} \Bigl) \Bigl\} \\ 
		& \times \int \operatorname{N}\left(\mu_{t, k} | \mu_0, \sigma_\mu^2\right)\exp\Bigl\{  \frac{1}{2} \Bigl(\sum_{i  \in \xi_k \cap [n_{\text{train}}] } \ddot{\mathcal{L}}_{\sigma}(\hat{\phi}_{i}) \Bigl) \mu_{t, k}^2 + \sum_{i \in \xi_k \cap [n_{\text{train}}] }\Bigl(\dot{\mathcal{L}}_{\sigma}(\hat{\phi}_{i}) -  \hat{\phi}_{t,i} \ddot{\mathcal{L}}_{\sigma}(\hat{\phi}_{i}) \Bigl) \mu_{t, k} \Bigl\} d \mu_{t, k} \\
		&=  \frac{1}{\sigma_{\mu} \sqrt{H(\xi) + \frac{1}{\sigma_{\mu}^2}}}
		\exp\Bigl\{ \sum_{i \in \xi_k \cap [n_{\text{train}}] } \Bigl( \mathcal{L}_{\sigma}(\hat{\phi}_{i})  - \hat{\phi}_{t,i}  \dot{\mathcal{L}}_{\sigma}(\hat{\phi}_{i}) +  \hat{\phi}^2_{t,i} \frac{\ddot{\mathcal{L}}_{\sigma}(\hat{\phi}_{i})}{2} \Bigl) - \frac{\mu_0^2}{2 \sigma_{\mu}^2} + \frac{\Bigl(J(\xi) + \frac{\mu_0}{\sigma_{\mu}^2}\Bigl)^2}{2 \left(H(\xi) + \frac{1}{\sigma_{\mu}^2}\right)}\Bigl\} 
	\end{aligned}
\end{equation*}
Similarly, based on this quadratic form, the posterior distribution of $\mu_{t,k}$ is:
\begin{equation*}
	\begin{aligned}
		\mu_{t,k} & | \cT_t, \vartheta_{-t}  \propto  \operatorname{N}\left(\mu_{t,k} | \mu_0, \sigma_\mu^2\right) \exp \Bigl\{ \sum_{i \in \xi_k \cap [n_{\text{train}}] } \mathcal{L}_{\sigma}(\phi_i-\hat{\phi}_i + \hat{\phi}_i)\Bigl\} \\
		& \approx \operatorname{N}\left(\mu_{t, k} | \mu_0, \sigma_\mu^2\right) \exp\Bigl\{ \sum_{i  \in \xi_k \cap [n_{\text{train}}] } \Bigl(\mathcal{L}_{\sigma}(\hat{\phi}_{i}) +   ( \mu_{t,k} - \hat{\phi}_{t,i} ) \dot{\mathcal{L}}_{\sigma}(\hat{\phi}_{i})  +   (\mu_{t,k} - \hat{\phi}_{t,i})^2  \frac{\ddot{\mathcal{L}}_{\sigma}(\hat{\phi}_{i})}{2} \Bigl) \Bigl\} \\
		& \propto \operatorname{N}\left(\mu_{t, k} | \mu_0, \sigma_\mu^2\right)\exp\Bigl\{  \frac{1}{2} \Bigl(\sum_{i  \in \xi_k \cap [n_{\text{train}}] } \ddot{\mathcal{L}}_{\sigma}(\hat{\phi}_{i}) \Bigl) \mu_{t, k}^2 + \sum_{i \in \xi_k \cap [n_{\text{train}}] }\Bigl(\dot{\mathcal{L}}_{\sigma}(\hat{\phi}_{i}) -  \hat{\phi}_{t,i} \ddot{\mathcal{L}}_{\sigma}(\hat{\phi}_{i}) \Bigl) \mu_{t, k}  \Bigl\}  \\ 
		& \propto \exp\Bigl\{  -\frac{1}{2} \Bigl(H(\xi_k)+ \frac{1}{\sigma_{\mu}^2} \Bigl) \mu_{t, k}^2 + \Bigl(J(\xi_k) + \frac{\mu_0}{\sigma_{\mu}^2}\Bigl) \mu_{t, k}  \Bigl\} \\
		\text{Therefore}&, \mu_{t,k} |  \cT_t, \vartheta_{-t} \sim \;  N \left(\frac{J(\xi_k) + \frac{\mu_0}{\sigma_{\mu}^2}}{H(\xi_k)+ \frac{1}{\sigma_{\mu}^2}},\; \frac{1}{H(\xi_k)+ \frac{1}{\sigma_{\mu}^2}} \right), \; \text{for }  k  = 1, \dots, \ell_t
	\end{aligned}
\end{equation*}
Meanwhile, the posterior distribution of $\sigma^2_{\mu}$ can be easily derived using conjugacy:
\begin{equation*}
	\begin{aligned}
		\sigma_{\mu}^2 &| \{\cT_{t}, \cM_{t}\}_{t=1}^T, \sigma  \propto \prod_{t=1}^{T} \prod_{k=1}^{\ell_t } \operatorname{N}\left(\mu_{t, k} | \mu_0, \sigma_\mu^2\right) \text{inv-Gamma} \left( \frac{a}{2}, \; \frac{b}{2} \right)  \\
		& \propto \left(\sigma_{\mu}^2\right)^{-\left( \frac{\sum_{t = 1}^{T}\ell_{t} + a}{2} +1\right)} \exp \Bigl\{ -\frac{\sum_{t = 1}^{T}\sum_{k = 1}^{\ell_t} \left( \mu_{t, k} - \mu_0\right)^2 + b}{2}\Bigl\} \\
		\text{Therefore, }
		\sigma_{\mu}^2 & | \{\cT_{t}, \cM_{t}\}_{t=1}^T, \sigma \sim\;  \text{inv-Gamma} \left( \frac{ \sum_{t = 1 }^T \ell_t + a}{2}, \; \frac{ \sum_{t = 1 }^T \sum_{k = 1}^{\ell_t}\left(\mu_{t,k} - \mu_0 \right)^2 + b}{2} \right)
	\end{aligned}
\end{equation*}

\subsection{Split and Merge Move Marginal Likelihood Ratios}\label{app_subsec:GSB-ratio}

In a split move, suppose a node $\xi$ is split into two offspring nodes $\left(\xi^*_{L}, \xi^*_{R} \right)$, then the log marginal likelihood ratio between $\calT^*$ and $\calT$ only depends on 
$\log m_{\xi^*_{L}} +\log m_{\xi^*_{R}}-\log m_{\xi}$ and can be expressed by
\begin{equation}\label{app_eq:split likelihood ratio}
	\begin{aligned}
		&\frac{1}{2} \left(\frac{\Bigl(J(\xi^*_{L}) + \frac{\mu_0}{\sigma_{\mu}^2}\Bigl)^2}{ H(\xi^*_{L}) + \frac{1}{\sigma_{\mu}^2}} + \frac{\Bigl(J(\xi^*_{R}) + \frac{\mu_0}{\sigma_{\mu}^2}\Bigl)^2}{ H(\xi^*_{R}) + \frac{1}{\sigma_{\mu}^2}}   -  \frac{\Bigl(J(\xi) + \frac{\mu_0}{\sigma_{\mu}^2}\Bigl)^2}{ H(\xi) + \frac{1}{\sigma_{\mu}^2}}   - \frac{\mu_0^2}{\sigma_{\mu}^2} \right) \\
		&- \frac{1}{2} \left( \log \Bigl(H(\xi^*_{L}) + \frac{1}{\sigma_{\mu}^2}\Bigl) 
		+  \log \Bigl(H(\xi^*_{R}) + \frac{1}{\sigma_{\mu}^2}\Bigl)-\log \Bigl(H(\xi) + \frac{1}{\sigma_{\mu}^2}\Bigl)\right) - \log \sigma_{\mu}. 	\end{aligned}
\end{equation}
Similarly, in a merge move, the log marginal likelihood ratio can be derived by reversing the split move. 
In a merge move, let $(\xi_{L}, \xi_{R})$  denote a pair of left and right leaf node sharing the same second-generation internal node in the current decision tree $\calT$ and $\xi^*$  denote the new leaf node 
by merging $(\xi_{L}, \xi_{R})$ in the proposed new tree $\calT^*$. The log marginal likelihood ratio of $\textbf{y}$ between $\calT^*$ and $\calT$ is
$\log m_{\xi^*}  - \log m_{\xi_{L}} - \log m_{\xi_{R}}$. Following the same derivation, this log marginal likelihood ratio has the form
\begin{equation}\label{app_eq:merge likelihood ratio}
	\begin{aligned}
		& \frac{1}{2} \left(\frac{\Bigl(J\left(\xi^{*}\right) + \frac{\mu_0}{\sigma_{\mu}^2}\Bigl)^2}{ H\left(\xi^{*}\right) + \frac{1}{\sigma_{\mu}^2}}  - \frac{\Bigl(J\left(\xi_{L}\right) + \frac{\mu_0}{\sigma_{\mu}^2}\Bigl)^2}{ H\left(\xi_{L}\right) + \frac{1}{\sigma_{\mu}^2}} - \frac{\Bigl(J\left(\xi_{R}\right) + \frac{\mu_0}{\sigma_{\mu}^2}\Bigl)^2}{ H\left(\xi_{R}\right) + \frac{1}{\sigma_{\mu}^2}} 
		+ \frac{\mu_0^2}{\sigma_{\mu}^2} \right) \\
		& - \frac{1}{2} \left( \log \Bigl(H\left(\xi^{*}\right) + \frac{1}{\sigma_{\mu}^2}\Bigl) - \log \Bigl(H\left(\xi_{L}\right) + \frac{1}{\sigma_{\mu}^2}\Bigl) 
		-  \log \Bigl(H\left(\xi_{R}\right) + \frac{1}{\sigma_{\mu}^2}\Bigl)\right) + \log \sigma_{\mu},
	\end{aligned}
\end{equation}

\subsection{Recursive Computation of Split Likelihood Ratios}\label{app_subsec:GSB-recur}

Alg.~\ref{app_alg:recur split} provides the detailed recursive procedure used to compute the log marginal likelihood ratios of all valid split moves on a single arborescence. This computation is required for each candidate arborescence in the candidate graph set $\dbmG$, and can be efficiently parallelized across arborescences to reduce overall cost. The algorithm returns both the likelihood ratio associated with each potential split edge and an indicator of its edge type (valid, invalid, or redundant). After obtaining the likelihood ratios for all valid splits, these values are combined with the corresponding merge ratios, prior ratios  to form the proposal probabilities and the weight function $\eta$ in the informed importance tempering (IIT) sampler.

\begin{algorithm}[!h]\setstretch{0.4}
	\caption{
		Recursive computations of the likelihood ratio of each possible split (SplitLikRatio) and determine edge type (ET) on a single arborescence for a given weak learner. }\label{app_alg:recur split}
	\begin{algorithmic} 
		{\small 
			\Require Arborescence $\dG=\{V,\dE\}$, weak learner $\calT$ with leaf node set $\underline{\xi} = \left(\xi_1, \ldots, \xi_{\ell} \right)$.  
			\State Set \Call{SplitLikRatio}{$\protect \de_v$} $\leftarrow \mathbf{-\infty}_{\ell}$, \Call{ET}{$\protect \de_v$} $\leftarrow \mathbf{1}_{\ell}$  $\quad \forall \de_v \in \dE$.
			\Procedure{ComputeSplitLikRatio}{$v$}
			\State $\underline{J}\left(\underline{\xi}_{R, v}^{*}\right) \leftarrow  \underline{J}(v) $,  $\underline{H}\left(\underline{\xi}_{R, v}^{*}\right) \leftarrow  \underline{H}\left(v\right) $, $ \underline{\xi}_{R, v}^{*} \leftarrow  v \cap \underline{\xi}$, $\operatorname{cnt} \leftarrow \mathbf{0}_{\ell}$ 
			\If { $v \notin  B(\dG)$} 
			\For {$v_c \in \operatorname{ch}(v)$}
			\State $\underline{J}\left(\underline{\xi}^*_{R, v_c}\right), \underline{H}\left(\underline{\xi}^*_{R, v_c}\right), 
			\underline{\xi}^*_{R, v_c} =$ \Call{ComputeSplitLikRatio}{$v_c$}
			\State $\underline{J}\left(\underline{\xi}_{R, v}^{*}\right) \leftarrow \underline{J}\left(\underline{\xi}_{R, v}^{*}\right) +  \underline{J}\left(\underline{\xi}_{R, v_c}^{*}\right) $, $\underline{H}\left(\underline{\xi}_{R, v}^{*}\right) \leftarrow \underline{H}\left(\underline{\xi}_{R, v}^{*}\right) +  \underline{H}\left(\underline{\xi}_{R, v_c}^{*}\right) $, 
			\State $\underline{\xi}_{R, v}^{*}  \leftarrow \underline{\xi}_{R, v}^{*} \cup \underline{\xi}_{R, v_c}^{*}$,  $\operatorname{cnt}\left[ | \underline{\xi}_{R, v_c}^{*} \cap  [n_{\text{train}}] |  > 0\right] += 1$ 
			\EndFor
			\EndIf
			\State $\underline{J}\left(\underline{\xi}_{L, v}^{*}\right) \leftarrow \underline{J}\left(\underline{\xi}\right) -  \underline{J}\left(\underline{\xi}_{R, v}^{*}\right) $, $\underline{H}\left(\underline{\xi}_{L, v}^{*}\right) \leftarrow \underline{H}\left(\underline{\xi}\right) -  \underline{H}\left(\underline{\xi}_{R, v}^{*}\right)$. 
			
			\For {$k \in 1:\ell$}
			\If {  $\left(| \xi_{k, R, v}^{*} \cap [n_{\text{train}}] |  = 0 \right) \text{\textbf{or}} \left( |\xi_k \cap [n_{\text{train}}]|- |\xi_{k, R, v}^{*}\cap [n_{\text{train}}] |  = 0 \right)$}
			\State 
			Set k$^\text{th}$ element of \Call{ET}{$\protect \de_v$} as $- 1$
			\ElsIf{$\left(\operatorname{cnt}\left[ k \right]=1\right)$ \text{\textbf{and}} $\left(| v \cap \xi_k \cap [n_{\text{train}}] | = 0\right)$ }
			\State  Set k$^\text{th}$ element of \Call{ET}{$\protect \de_v$} as $0$
			\EndIf
			\EndFor
			\State Calculate \Call{SplitLikRatio}{$\protect \de_v$} $\left[\right.$  \Call{ET}{$\protect \de_v$} $ = 1 \left.\right]$ using equation~\eqref{app_eq:split likelihood ratio} in Supplementary
			\EndProcedure
		}
		\State \Call{ComputeSplitLikRatio}{$v_0$}. \\
		\Return \Call{SplitLikRatio}{$\protect \de_v$}, \Call{ET}{$\protect \de_v$} $\quad \forall \de_v \in E$
	\end{algorithmic}
\end{algorithm}

\subsection{Maintaining Graph Connectivity}\label{app_subsec:graph-int}

Figure \ref{app_fig:graph integrity} demonstrates why we do not remove edges of arborescences in split moves. Mainly, edge removal breaks the graph connectivity which may lead to incorrect marginal likelihood calculation and label updating. To be specific,  the first column of the blue bracket shows the cluster memberships after splitting root node by $\de_3$ of $\dG_1$ with or without removing this edge. We further split $\xi_2$ by $\de_8$ of $\dG_1$ (shown in the second column of blue bracket). Removing this edge will result in $\dG_1$ becoming two sub-graphs denoted as $\dG_{1,1}$ and $\dG_{1,2}$ with root nodes being that of the original graph $v_4$ and node $v_8$, respectively. Suppose we split another leaf node $\xi_1$ with respect to $\dG_1$. Calculating the split likelihood on the entire $\dG_1$ differs from the result of computing it separately on two sub-graphs; $\operatorname{decs}(v_{12})$ are different on $\dG_1$ and $\dG_{1,1}$ which lead to different $\xi^{*}_{R,1} = \xi_1 \cap [n_{\text{train}}] \cap \left( v_{12} \cup \operatorname{decs}(v_{12}) \right)$, split likelihood ratio split by $v_{12}$ and label updating (shown in the red bracket). 
\begin{center}
	\begin{figure}[h]
		\centering
		\includegraphics[width=6in,height=4in]{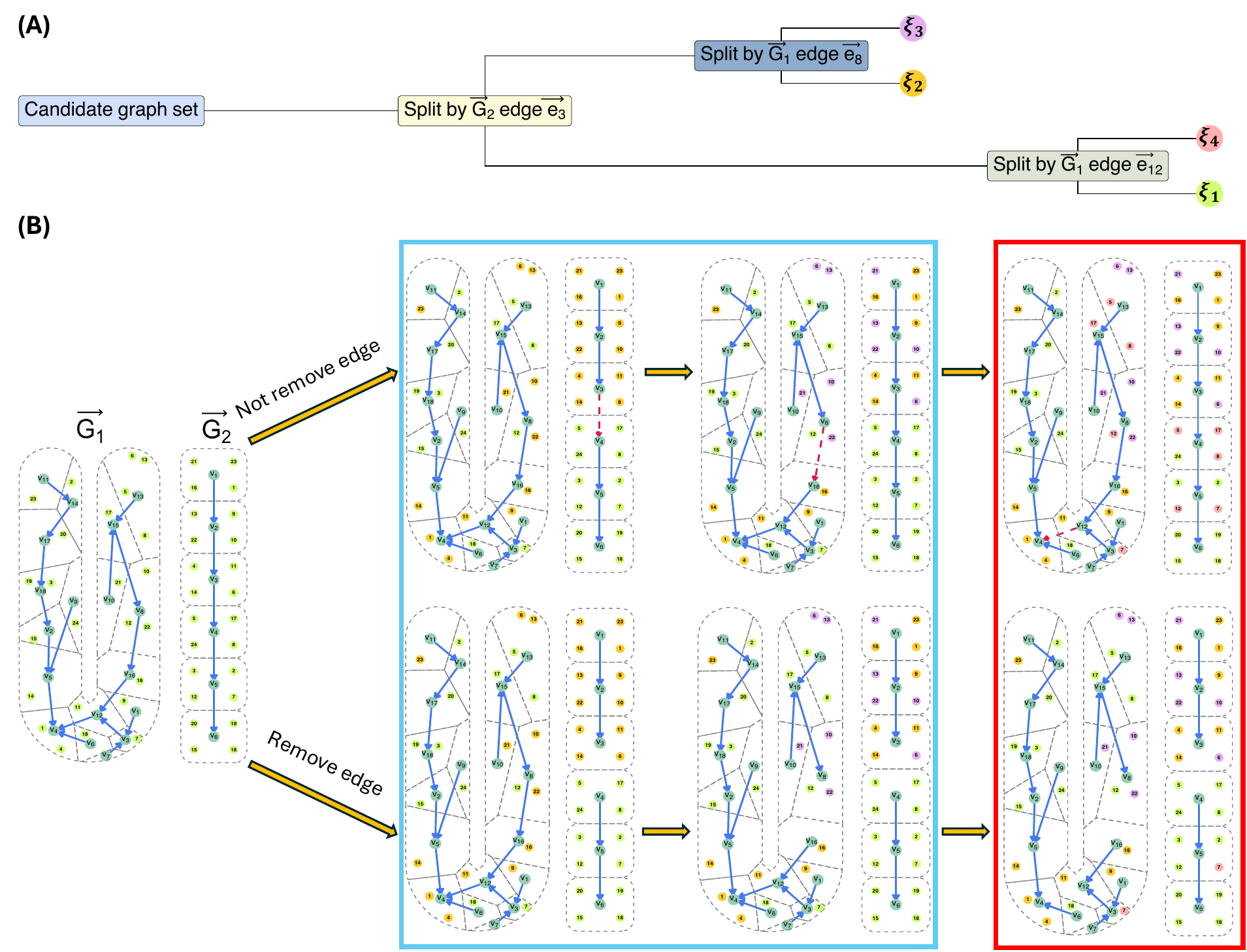}
		\setlength{\belowcaptionskip}{-20pt}
		\caption{
			(A): A graph split decision tree with three internal nodes and four leaf nodes; \\
			(B): The first and second rows denote the label updating procedure without and with edge removal under the same tree growth decision rule. The color of each observation represents the cluster it belongs to.} 
	\label{app_fig:graph integrity}
\end{figure}
\end{center}

\subsection{Further Discussion of Equivalent Edge Set}\label{app_subsec:redund-egde}

\noindent\textbf{Redundant edges determination.}
To determine whether an edge $\de_v$ is redundant with respect to a leaf node $\xi$, we evaluate the following conditions: (1) $| v \cap \xi \cap [n_{\text{train}}] | = 0$; (2) Among the child vertices of $v$, there exists only one vertex, $v_c$, such that $\de_{v_c}$ is a valid or redundant split with respect to $\xi$. In other words, splitting by $\de_{v_c}$ results in a bi-partition of $\xi$, regardless of whether this bi-partition is unique. We use the vector `$\operatorname{cnt}$' to track the number of such child vertices for $v$ with respect to each leaf node in Algorithm~\ref{app_alg:recur split}.

\medskip
\noindent\textbf{Equivalent edge set random sampling.}
As shown in Figure \ref{app_fig:alg demo}, after applying Alg.~\ref{app_alg:recur split} to an aborescences, we obtain a SplitLikRatio matrix and an Edge type matrix for each candidate graph, both with dimensions $|V| \times \ell$. 
Suppose $\de_v$ is selected by the IIT algorithm to split the leaf node $\xi$, then we check the edge type of $\de_{\operatorname{pa}(v)}$ with respect to $\xi$. If the edge type is `1' or `-1', then we conduct this split. Otherwise, if the edge type is `0', which means $\de_v$ is the bottom edge of an equivalent edge set, we go through the edge type of $\de_{\operatorname{ances}(v)}$ in bottom-up order until we meet `-1' or `1'. All the edges being gone through except the last one together with $\de_v$ consist of an equivalent edge set of $\xi$. Once we identify all the edges belonging to this equivalent edge set, we can easily sample one to form the split rule. We demonstrate a concrete example in Figure \ref{app_fig:alg demo}. Suppose $\de_{11}$ of $\dG_{1}$ is selected to split leaf node $\xi_2$ by. We check the edge type of $\de_{14}$ with regard to $\xi_2$ which is `0'. Following the method mentioned above, we get $\{\de_{11}, \de_{14}, \de_{17},\de_{18},\de_{2} \}$ consisting of an equivalent edge set for $\xi_2$. We can randomly sample one edge, say $\de_{18}$, from this set to make the split. 
We then check the edge type of $\de_{2}$ with regard to $\xi_1$, which is `1', leading to a split of $\xi_1$ by $\de_{18}$. 

\begin{figure}[ht]
\centering
\includegraphics[width=5.4in,height=2.7in]{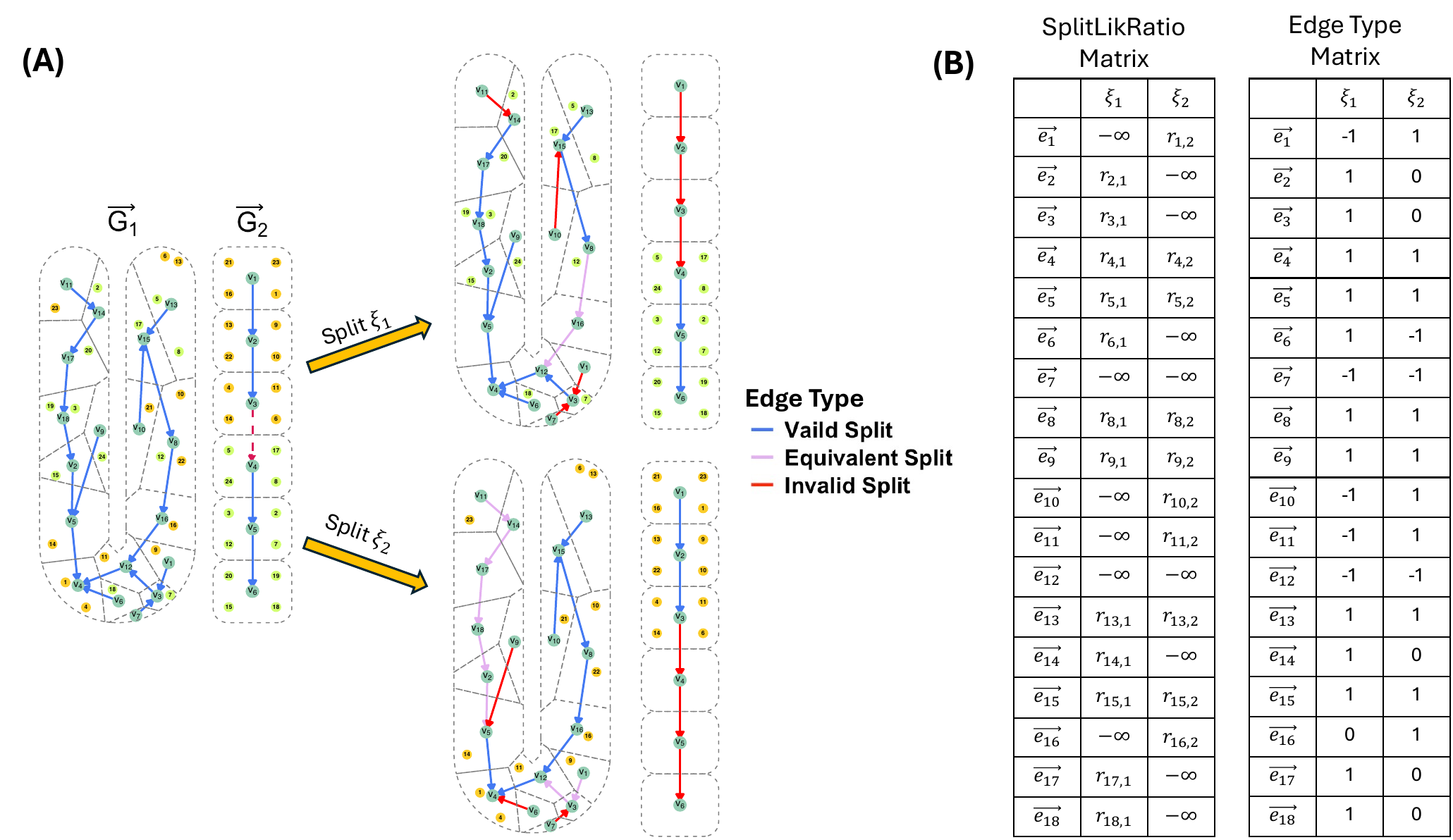}
\caption{(A): The decision tree after splitting the root node into two leaf nodes $\xi_1$ (lime dots) and $\xi_2$ (orange dots) by the edge in $\protect \dG_2$ marked by the red dashed line. The edge types when $\xi_1$ or $\xi_2$ is selected for splitting are marked by different edge colors; \\
	(B): The SplitLikRatio and Edge type matrices for $\protect \dG_1$ after splitting the root node by an edge in $\protect \dG_2$ following Alg.~\ref{app_alg:recur split}. In the SplitLikRatio matrix, $r_{i,j}$ represents splitting leaf node $\xi_{j}$ by $\protect \de_i$ of $\protect \dG_1$. In the Edge Type Matrix, we use 1', -1', 0' to represent valid, invalid, and redundant edges, respectively. Notice that $v_{12}$ is the root vertex of $\protect \dG_1$, so $\protect \de_{12}$ does not exist and is marked as -1' in the Edge type matrix. }
\label{app_fig:alg demo}
\end{figure}

\section{Simulation Data Generation Settings}\label{app_sec:exp-setting}
For different response data, the likelihood function, the priors of model parameters, and hyper-parameters settings are different. The detailed settings for each type of response data are as follows: \\
\noindent\underline{Continuous data:} We use the Gaussian likelihood for continuous response data. The log marginal likelihood function has an explicit quadratic form and the corresponding $\dot{\mathcal{L}}_{\sigma}(\phi)$ and $\ddot{\mathcal{L}}_{\sigma}(\phi)$ are shown in Section~\ref{subsec:quadratic}. 
For the residual variance, $\sigma^2$, we follow standard BART and assume
$\begin{aligned}
\sigma^2  \sim \text{inv-Gamma} \left( \frac{\nu}{2}, \frac{\nu \lambda}{2}\right)
\end{aligned}$, 
where the value of $\nu$ is 3 and we pick the value of $\lambda$ so that the 0.9 quantile of the prior on $\sigma$ is located at the sample standard deviation $\hat{\sigma}$ of the response data, that is, $P(\sigma<\hat{\sigma})=0.9$. Under this setting, the posterior of $\sigma^2$ follows $\text{inv-Gamma} \left( \frac{\nu + \sum_{i= 1}^{n}(y_i - \hat{y_i})^2}{2}, \frac{\nu \lambda + n}{2}\right)$. We also follow BART to rescale and recenter the response data such that it ranges from -0.5 to 0.5 before fitting GS-BART model. We set the initial value of $\sigma_\mu$ at $0.5/2 \sqrt{T}$.
\\
\underline{Count data:} We employ the variance model for count response data. Specifically, we assume $y_i  \sim \text{Normal}\left\{ \phi_i , \phi_i\right\}$. 
In this case, $\dot{\mathcal{L}}_{\sigma}(\phi_i)$ and $\ddot{\mathcal{L}}_{\sigma}(\phi_i)$ have the form:
$$
\begin{aligned}
\dot{\mathcal{L}}_{\sigma}(\phi_i) &= -0.5 + \frac{\left(y_i - \exp \left( \phi_i \right)\right)^2}{2 \exp \left( \phi_i \right)} + y_i - \exp \left( \phi_i \right), \quad 
\ddot{\mathcal{L}}_{\sigma}(\phi_i) &= -y_i - \frac{\left(y_i - \exp \left( \phi_i \right)\right)^2}{2 \exp \left( \phi_i \right)}
\end{aligned}
$$ 
The initial value of $\sigma_\mu$ is $\sigma_\mu = 6/k \sqrt{T}$ with proper $k$ with the rationale that $-6$ to $6$ covers the typical ranges of log response values in our examples.
\\
\underline{Classification data:} Suppose our response data has $C$ unordered categories. Let $y_{i,j}$ denote a binary variable indicating whether the $i$-th observation belongs to the $j$-th category. We fit $C$ binary logistic GS-BART regression models where the probability that a sample belongs to a particular category $j$ versus all other categories is modeled by $1/\left(1+\exp(-\phi_{i,j})\right)$, for $j=1,\ldots, C$. 
For each binary logistic regression, $\dot{\mathcal{L}}_{\sigma}(\phi_{i,j})$ and $\ddot{\mathcal{L}}_{\sigma}(\phi_{i,j})$ have the form:
$$
\begin{aligned}
\dot{\mathcal{L}}_{\sigma}(\phi_{i,j}) &= y_{i,j}  - \frac{\exp(\phi_{i,j})}{1+\exp(\phi_{i,j})}, \quad 
\ddot{\mathcal{L}}_{\sigma}(\phi_{i,j}) &= - \frac{\exp(\phi_{i,j})}{(1+\exp(\phi_{i,j}))^2}
\end{aligned}
$$
The initial value of $\sigma_\mu$ is adopted as $\sigma_\mu = 3/k \sqrt{T}$ since the interval $\left[-3, 3\right]$ for $\operatorname{logit}(p_j)$ covers the majority of possible values for $p_j$.
The Bayesian estimates, $\hat{\phi}_{i, j}$, are then used to  predict the probability of $\mathbb{P}(y_{i}=j)$ using  $\hat{p}_{i,j} = \frac{\exp\{ \hat{\phi}_{i, j} \}}{\sum_{j = 1}^{C}\exp\{ \hat{\phi}_{i, j} \}}$.

\section{Synthetic and Real Data Sets Description}\label{app_sec:data-intro}
We provide a detailed description of synthetic and real datasets and the sources of 
real datasets in this section. \\
\noindent\underline{U-shape example:} We generate synthetic data at $n_{\text{train}} = 800$ training locations and $n_{\text{test}}=200$ testing locations sampled from the two-dimensional U-shaped domain $\calS$ divided into three true subregions as shown in Figure \ref{fig:Ushape results}(A). We generate 5 unstructured features: $(x_1, \cdots, x_5)$. The true function $\phi$ depends on $(\bfs, x_1)$ and their interaction and is piecewise smooth within each subregion but has sharp changes across boundaries. For continuous response data, we generate from a normal error model centered at the true function and with $\sigma = 0.15$. For count response data, we generate data from Poisson model 
with the rate function $r_0(\phi) = \lfloor\exp(\frac{\phi + 7}{3})\rceil$. 
For classification data, we generate the data by discretizing $\phi$ into five categories using 0.2, 0.3, 0.65, and 0.8 quantiles as the thresholds. 

\noindent\underline{Torus example:} We randomly sample $n_{\text{train}} = 500$ training locations and $n_{\text{test}} = 200$ testing location on $\calS$ which is a two-dimensional bitten torus manifold embedded in $\mathbb{R}^3$ \citep{niu2019intrinsic}. We generate two unstructured features together with three coordinates of locations $\bfs$ to construct chain graphs. The true mean function only depends on $(\bfs, x_1)$ and has two discontinuities. 
Similar to the U-shape example, we generate the continuous response data by adding a normal error to the true function with $\sigma = 0.1$. For count response data, we generate data from Poisson with 
the rate function: $r_0(\phi) = \lfloor \exp(\frac{\phi + 17}{8}) )\rceil$. We generate the data for classification tasks by discretizing $\phi$ into five categories using the same threshold as the U-shape example. 

\noindent\underline{Friedman example:} We include Friedman function \citep{friedman1991multivariate} as an example involving only unstructured features. We randomly generate $n_{\text{train}} = 400$ training data and $n_{\text{test}} = 100$ testing data with $5$ unstructured covariates: $(x_1, \cdots, x_5)$. The true function is $10 \sin \left(\pi x_1 x_2\right)+20\left(x_3-0.5\right)^2+10 x_4$. We generate data from the continuous response regression model with $\sigma = 1$.

\noindent\underline{King House:} We analyze the home sales price data of $n = 20149$ homes in King County, Seattle, Washington, which is obtained from https://geodacenter.github.io/data-and-lab/\\KingCounty-HouseSales2015/. We took the logarithm of the house sales price ( in U.S. dollars) as the response and fit the model using the house's latitude and longitude as the spatial structured features together with the other $p = 15$ unstructured features.

\noindent\underline{NYC Education:} 
The datasets we analyze are New York City income \& education (n = 1690, p = 6). 
The response variable is the income per capita of New York City residents. Location is included as the spatial structured feature, while educational attainment level and the other four variables are treated as unstructured features in the regression model. Data source: https://geodacenter.github.io/data-and-lab//NYC-Census-2000. 

\noindent\underline{US Election:} regression using social demographic data to explain the (logit-transformed) proportion of Republican votes in the 2016 presidential elections at the county level. We limit our analysis to counties within the 48 contiguous United States with a spatial adjacency graph with $3071$ nodes, $8669$ edges, and $p=55$ regular features. The dataset is obtained from:https://github.com/tonmcg/US\_County\_Level\_Election\_Results\_08-20 

\noindent\underline{US Cancer:} We analyze the cancer counts of  $n = 3103$ counties in Unite States. We take the population of each county as our offset variable in the Poisson regression, and we use some social demographic variables together with the latitude and longitude data (p = 29) to generate chain arborescences to fit our model. \\
Data source: https://ghdx.healthdata.org/record/ihme-data/united-states-cancer-mortality-rates-county-1980-2014

\noindent\underline{Pollution:} We analyze the county-level AQS data in 2023 from: https://www.epa.gov/air-trends/air-quality-cities-and-counties. We consider a binary indicator of whether polluted days were observed in each county or not. Only 893 counties have measurements. 
We use $p = 29$ features to construct chain graphs, including both latitude and longitude data of the centroid location of each county together with social economical variables. In order to maintain the connectivity of the structure graphs, we use the adjacency graph of all counties of the contiguous United States to generate tree arborescences. 

\noindent\underline{Flood:} The response data is a binary variable used to represent whether a flood occurs in the county or not. We have observations at $n = 3037$ counties. We use the same structured graph and unstructured features as the Pollution data example. 

\noindent\underline{Cora Network:} The citation network data has $n =  2780$ vertices, each representing a scientific publication,  5429 edges representing the citation relationships between papers, and $p = 1433$ binary features indicating the presence or absence of specific words from a dictionary. The response variable is categorical, specifying one of seven research topics to which each publication belongs. We use all binary features together with 5 tree arborescences generated from this citation network to construct the candidate graph sets to train each weak learner of GS-BART. Data source: https://fit.cvut.cz/cs/veda-a-vyzkum/cemu-se-venujeme/projekty/relational.

Table \ref{Tab:experiment setting} provides the number of arborescences and their vertices used for each spatial or network-structured and unstructured feature for each data example.

\begin{table}[h]
\small
\begin{center}
	\begin{threeparttable}
		\begin{tabular}{lcc}
			\toprule[1pt]
			DATASETS & $\mathbf{|\dbmG_\text{structured}|}/\mathbf{|\dbmG_\text{unstructured}|}$ & $\mathbf{|V_\text{structured}|}/\mathbf{|V_\text{unstructured}|}$ \\
			\hline  
			U-SHAPE & 5/7 & 100/100  \\
			TORUS & 5/8 & 100/100    \\ 
			FRIEDMAN  & 0/5  & -/100  \\
			\hline 
			NYC Education & 5/7 & 100/100   \\
			King County House  & 5/17 & 245/100   \\ 
			US Election  & 5/53  & 100/100  \\
			US Cancer & 5/31  & 100/100 \\
			Air Pollution & 5/29  & 100/100   \\ 
			Flood & 5/29 & 100/100   \\ 
			Cora Network & 5/1433 & 100/2 \\
			\bottomrule[1pt]
		\end{tabular}
		\setlength{\belowcaptionskip}{-10pt} 
		\caption{The numbers of arborescences  and vertices used for the spatial or network-structured  
			and unstructured features for each experiment.}\label{Tab:experiment setting}
	\end{threeparttable}
\end{center}
\end{table}



\section{Additional Results for Numerical Experiments}\label{app_sec:add-res}

\subsection{Mixing Performance}\label{app_subsec:GSB-mixing}  

We compare the mixing effectiveness of the IIT sampler in GS-BART with the standard split–merge sampler in BART using both the U-shape and Friedman normal regression examples. 
Figure~\ref{app_fig:ushapetrace} displays the trace plots of the mean squared prediction error after 15 sweeps for GS-BART and after 5000 sweeps for BART, as well as scatter plots comparing the prediction performance from GS-BART and BART for one U-shaped simulation. GS-BART exhibits improved mixing performance with faster convergence and a higher effective sample size (ESS) rate, despite using substantially fewer sweeps. Its predictive accuracy also improves considerably, with MSPE reduced to only 34\% of that of BART. Additionally, we evaluate the 95\% credible interval coverage at each location. GS-BART achieves a coverage rate of 92\%, whereas BART severely underestimates uncertainty, with a coverage rate of only 64\%. The superior performance of GS-BART over BART is potentially attributed to both its ability to efficiently handle graph-structured features and the efficiency of IIT over the standard split-merge sampler. 

To examine the sole contribution of the IIT algorithm for mixing, we consider the classic Friedman example without giving GS-BART the advantage of considering spatial and network features. All methods use the same covariates and the same number of trees, and GS-BART adopts the same candidate cut points as BART to construct chain graphs.  The MSPE trace plot in Figure~\ref{app_fig:friedmantrace} shows that GS-BART converges very quickly and mixes well after only 15 burn-in sweeps. Consistent with the U-shape example, GS-BART also achieves higher ESS and improved coverage rates. Moreover, GS-BART provides more accurate predictions on held-out responses than BART. This improvement is statistically significant, as confirmed by the prediction comparisons with repeated experiments summarized in Table 1. These findings suggest that the IIT algorithm achieves improved mixing efficiency compared to standard split-merge approaches, which may in turn contribute to its superior predictive performance. 


We also compare the mixing performance of GS-BART and the generalized BART algorithm in~\cite{linero2024generalized} on the U.S. cancer dataset, a count regression problem. As shown in Figure~\ref{app_fig:CancerTrace}, GS-BART demonstrates markedly improved mixing behavior over generalized BART for count regression, with more stationary trace plots, substantially larger effective sample sizes, faster coverage, and higher predictive accuracy.

\begin{figure}[h]
\begin{center}
	\includegraphics[scale=.25]{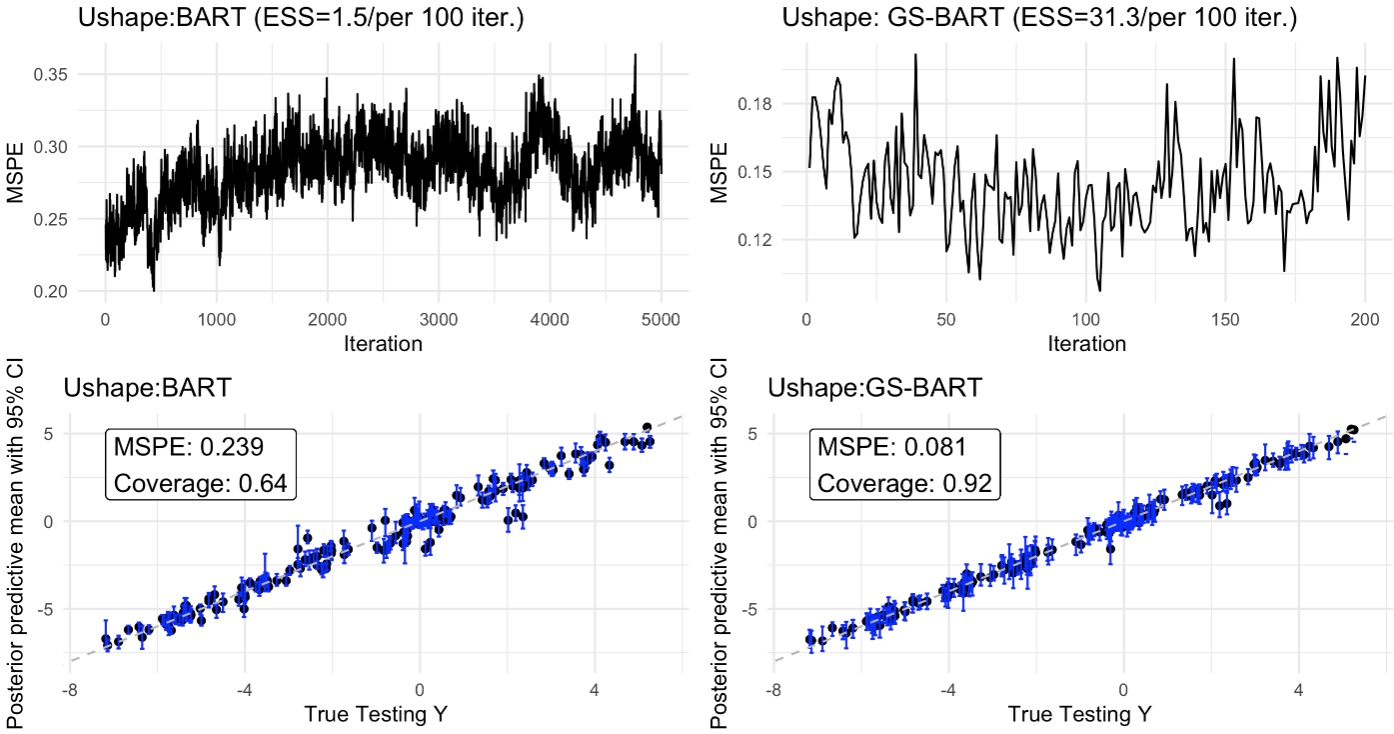}
	\setlength{\belowcaptionskip}{-15pt}
	\caption{U-shape example: Trace plots of MSPE for predictive posterior samples and scatter plots of predictions versus true test data. Titles report the ESS rate; boxes display the MSPE of the posterior mean and the coverage rate.}
	\label{app_fig:ushapetrace}  
\end{center}	
\end{figure}

\begin{figure}[h]
\begin{center}
	\includegraphics[scale=.25]{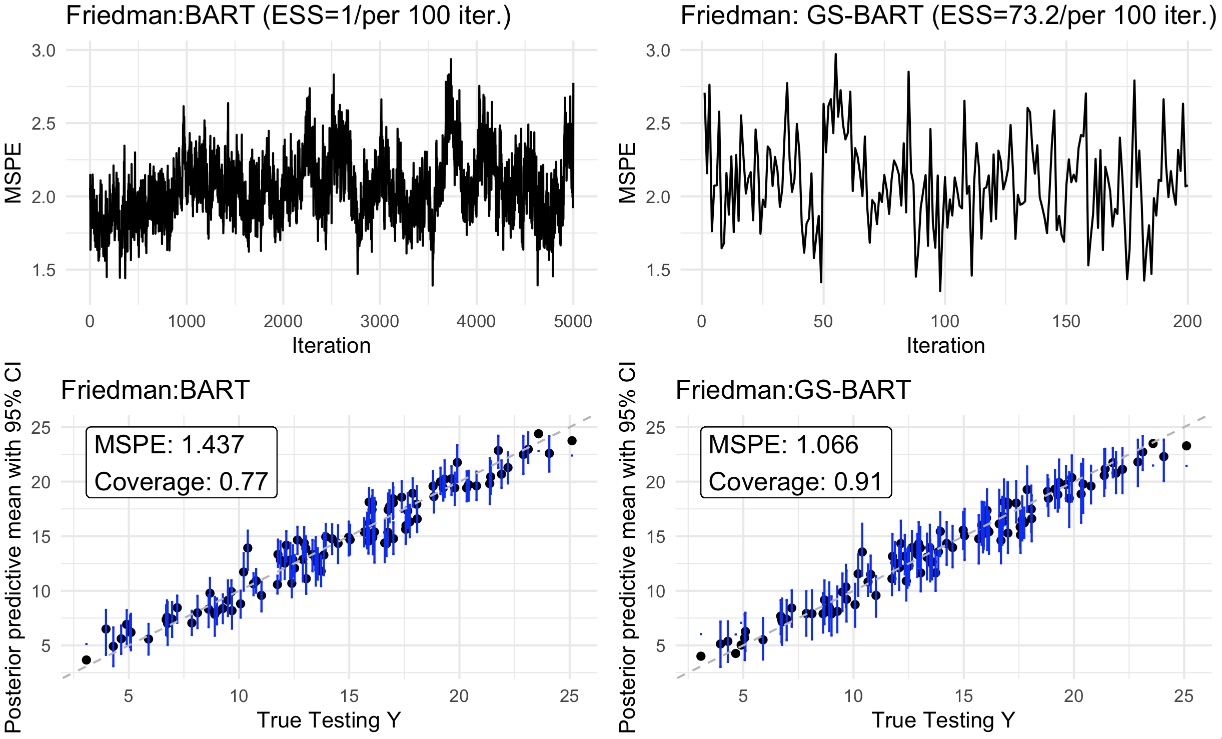}
	\setlength{\belowcaptionskip}{-15pt}
	\caption{Friedman example: Trace plots of MSPE for predictive posterior samples and scatter plots of predictions versus true test data. Titles report the ESS rate; boxes display the MSPE of the posterior mean and the coverage rate.}
	\label{app_fig:friedmantrace}   
\end{center}	
\end{figure}

\begin{figure}[h]
\begin{center}
	\includegraphics[scale=.45]{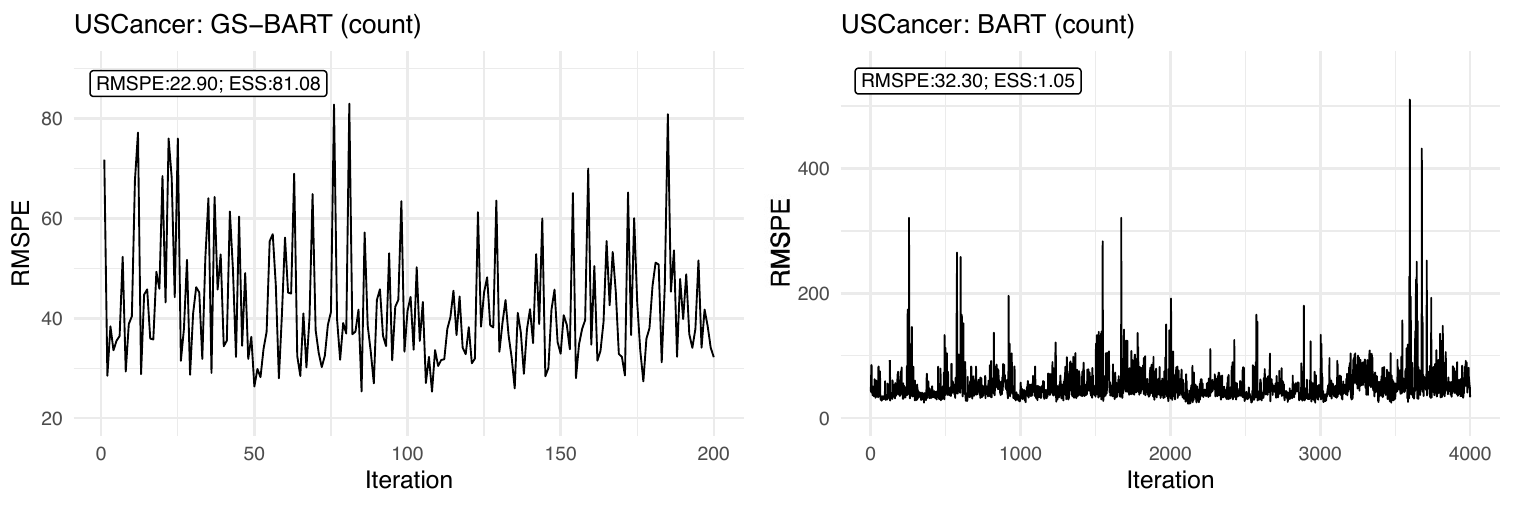}
	\caption{Trace plots of RMSPE of each posterior sample, RMSPE of the posterior mean, and ESS per 100 iterations for the US cancer count data using GS-BART and generalized BART for count regression.}\label{app_fig:CancerTrace}
\end{center}	
\end{figure}

\subsection{Additional Results of the U-shape Example}\label{app_subsec:Ushape-add}

\noindent\textbf{Sensitivity analysis.}
GS-BART involves the choice of model tuning parameters that can influence its performance.  We mainly focus on studying the effects of those parameters specific to GS-BART or highly relevant to the computational complexity, omitting the sensitivity analysis for most parameters shared with BART since we adopt the same settings with BART. Specifically, we focus on three parameters: the number of weak learners $T$, the number of arborescences sampled from each structural graph, $M$, and the number of Voronoi centroids used to construct the structural candidate graph,  $|V|$. 

The sensitivity analysis is performed on the U-shape normal regression example with $n_{train} = 800$ and $n_{test} = 200$. Since the arborescences and Voronoi centroids are randomly picked, we repeat the experiment 50 times and show the averaged results. 

We first study the effects of $T$ and $M$ while fixing $|V| = 100$. The MSPEs using different combinations of $M$ and $T$ are reported in Table \ref{app_Tab: Ushape Sensitivity analysis 1}. 
We notice that increasing $T$ and $M$ generally improves the out-of-sample performance for a fixed $|V|$. However, beyond a certain threshold, there appear to be diminishing or negative returns in prediction accuracy. Next, we fix $M = 5$ and $T = 50$ and test the performance for different numbers of vertices $|V|$ of candidate arborescences. Table~\ref{app_Tab: Ushape Sensitivity analysis 2} presents the MSPE across different values of $|V|$. While increasing the number of vertices typically enhances prediction accuracy, it also incurs higher computational costs. Beyond a certain point, further increasing $|V|$ yields little additional improvement.

In summary, GS-BART demonstrates relatively robust performance to the choices of $T$, $M$, and $|V|$, as long as they are sufficiently large. Based on the sensitive analysis results, we used $T = 50$, $M = 5$, and $|V|=100$ as the default setting in most of our numerical studies, and we found they generally achieve a good balance between prediction accuracy and computational efficiency. Nonetheless, these hyperparameters can be further tuned to enhance the performance of GS-BART in specific applications. 
One notable advantage of ensemble tree methods lies in their ability to leverage the diversity among weak learners. Unlike traditional models, where each model must be highly accurate on its own, ensemble tree methods can achieve strong predictive performance by aggregating the outputs of multiple weak learners. As long as there is a sufficient number of multiple weak leaders, the randomness involved in constructing arborescences helps enrich the diversity of potential covariate space partitions to better capture the underlying function variation patterns.

\begin{table}[h]
\centering
\small
\begin{threeparttable}
\begin{tabular}{lcccccc}
	\hline & $M=1$ & $M=3$ & $M=5$ & $M=10$ & $M=30$ & $M=50$ \\
	\hline
	$T=10$  & 1.13 (0.21) & 1.18 (0.26) & 1.14 (0.23) & 1.13 (0.15) & 1.16 (0.22) & 1.24 (0.26) \\
	$T=25$  & 0.95 (0.18) & 0.92 (0.19) & 0.88 (0.20) & 0.83 (0.16) & 0.89 (0.20) & 0.91 (0.18) \\
	$T=50$  & 0.84 (0.19) & 0.89 (0.16) & 0.82 (0.22) & 0.82 (0.16) & 0.91 (0.14) & 0.88 (0.13) \\
	$T=75$  & 0.87 (0.16) & 0.90 (0.18) & 0.84 (0.16) & 0.83 (0.15) & 0.84 (0.10) & 0.84 (0.12) \\
	$T=100$ & 0.92 (0.15) & 0.88 (0.14) & $\mathbf{0.79}$ (0.15) & 0.81 (0.13) & 0.84 (0.11) & 0.87 (0.13) \\
	\hline
\end{tabular}
\caption{Sensitivity analysis: the averaged MSPE $\times 10$ (sd $\times 10$).}
\label{app_Tab: Ushape Sensitivity analysis 1}
\end{threeparttable}
\end{table}


\begin{table}[h]
\centering
\footnotesize
\begin{threeparttable}
\begin{tabular}{lccccccc}
	\hline & & $|V|=10$ & $|V|=25$ & $|V|=50$ & $|V|=75$ & $|V|=100$  & $|V|=200$ \\
	\hline 
	\multicolumn{2}{c}{$\text{MSPE} \times 10 \; (\text{sd} \times 10)$} 
	& 1.05 (0.16) & 1.02 (0.22) & 0.97 (0.16) & 0.90 (0.12) & $\mathbf{0.82}$ (0.22) & 0.87 (0.23) \\
	\multicolumn{2}{c}{$\text{MAPE} \times 10 \; (\text{sd} \times 10)$} 
	& 2.26 (0.13) & 2.23 (0.16) & 2.15 (0.14) & 2.10 (0.11) & $\mathbf{2.02}$ (0.15) & 2.05 (0.15) \\
	\multicolumn{2}{c}{$\text{Time}$} 
	& 69.51 & 94.44 & 141.91 & 198.22 & 250.01 & 544.29 \\
	\hline
\end{tabular}
\caption{Sensitivity analysis: the averaged MSPE $\times 10$, MAPE $\times 10$, and running time for different $|V|$.}
\label{app_Tab: Ushape Sensitivity analysis 2}
\end{threeparttable}
\end{table}


\noindent\textbf{Absolute error surface.}
The predictive absolute error surfaces of different methods are demonstrated in Figure \ref{app_Tab:Ushape Absolute error}. Compared with the competing methods, GS-BART has a smaller absolute error in most areas of the U-shape region. 

\begin{center}
\begin{figure}[h]
\centering
\includegraphics[scale=.4]{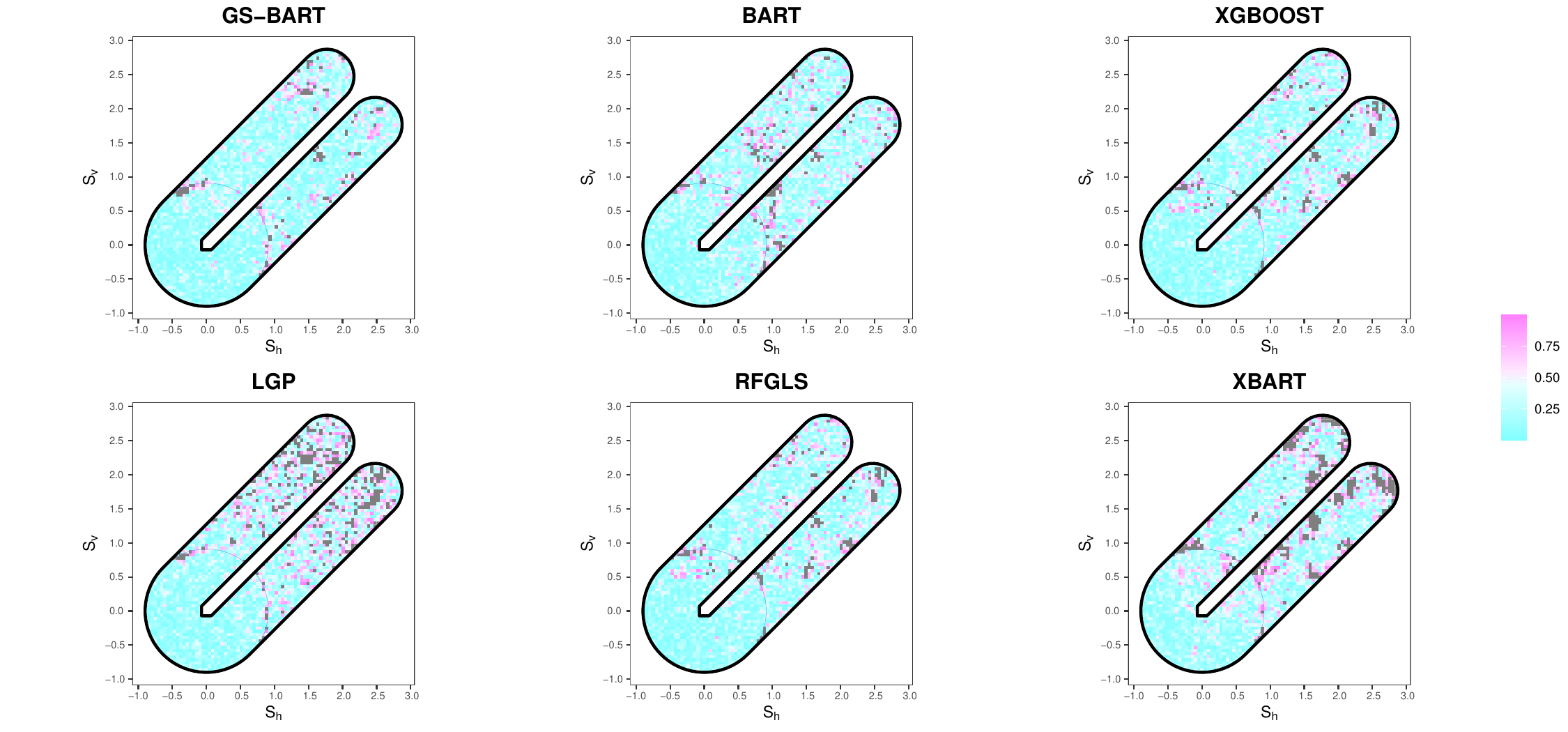}
\caption{Absolute error surfaces of various methods for the U-shape simulation example in the setting of $n=800, p=2$ and $\sigma=0.15$. Absolute error larger than 1 is marked by gray. Steel blue circles indicate discontinuity surfaces in the true function projected to $\calM$. }
\label{app_Tab:Ushape Absolute error}
\end{figure}
\end{center}


\subsection{Additional Results of the King County House Sales Example}\label{app_subsec:KingHouse-add}
We pick the living area (square feet) as an example of unstructured variable partial effect analysis since it has the highest variable importance among all the unstructured continuous variables. We choose six representative locations and show their corresponding living area partial effect plots in Figure \ref{app_fig:living area partial dependence}. Overall, the house sales price increases as the living area grows in all six representative locations. There is a significant jump in the estimated selling price of a home when it exceeds 5,000 square feet of living space, after which the trend of increasing selling price with growing living space levels off. This phenomenon may be explained by the fact that houses with more than 5,000 square feet of living space are generally considered to be more luxurious, and the cost of decorating such homes is higher, making them more expensive to sell. However, as the square footage of a home continuous increases, the marginal utility in the living experience diminishes, while the cost of repairs and maintenance continues to increase, resulting in no further increase in the sales price.

\begin{center}
\begin{figure}[h]
\centering
\includegraphics[scale=.45]{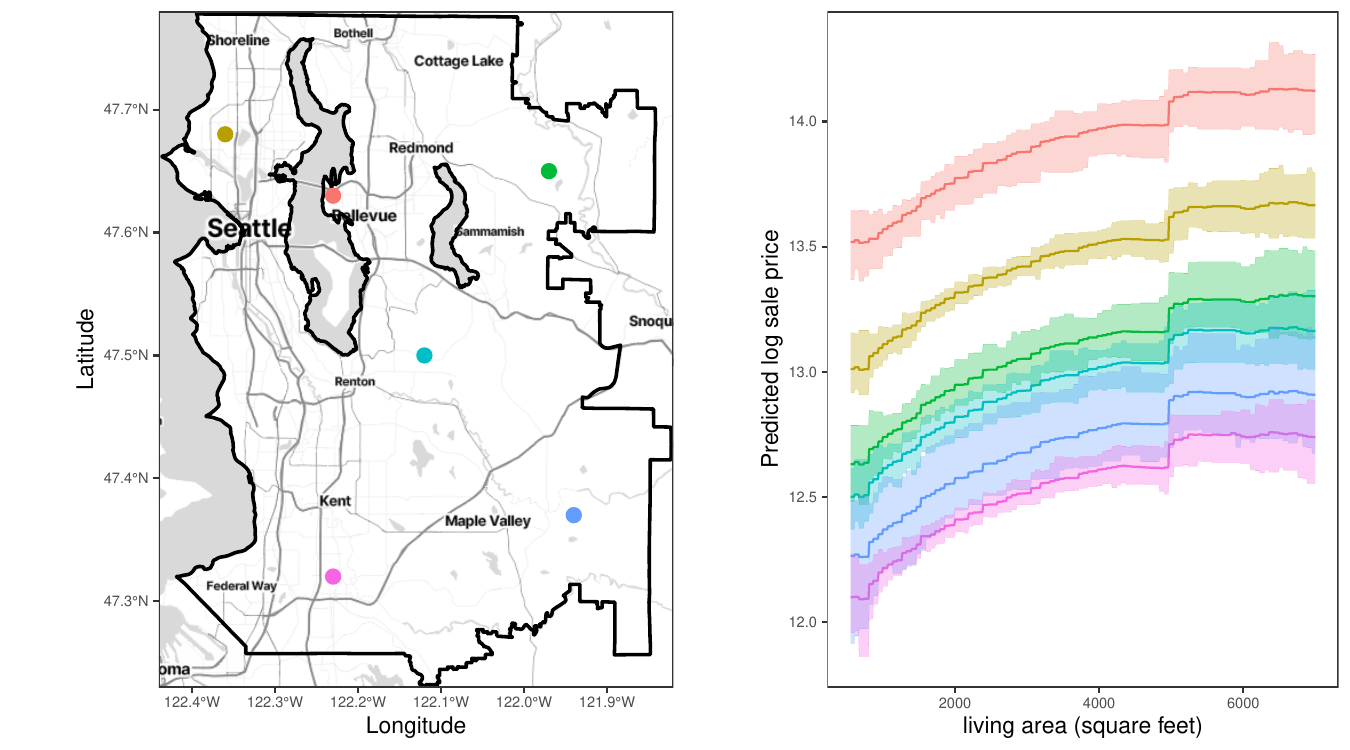}
\caption{Six representative locations are selected. Six lines represents predicted log house sales price versus living area. The shadowed ribbons represents the 95\% predictive credible intervals of six representative locations.}
\label{app_fig:living area partial dependence}
\end{figure}
\end{center}

\section{Additional Experiments on Column Subsampling}\label{app_sec:col-sub}

Column subsampling is a popular strategy to handle high-dimensional features in ensemble decision tree methods such as XGBoost and Random Forest, where each decision tree is trained only on a random subset of features. It is convenient to implement similar column subsampling strategies in GS-BART, by viewing them as particular choices of user-defined candidate graph sets. For XGBoost and Random Forests, it is generally advised to apply uniform random feature subsampling with a rate in the range of 0.5 to 0.9, unless prior knowledge suggests otherwise. 

Table~\cref{app_Tab:subcol} presents results on the U-shape dataset with $p = 12$ predictors, where models were trained
with varying numbers of randomly selected features per tree. The results indicate that the performance of 
GS-BART remains robust once the subsampling rate exceeds 0.5, with runtime decreasing
as fewer features are included at the cost of only a mild loss in accuracy.

\begin{table}[h!]
\centering
\begin{tabular}{cccc}
\toprule[1pt]
\# of Sub Cols & MSPE $\times 10$  (sd $\times 10$  )  & MAPE $\times 10$  (sd $\times 10$ ) & Time (Sec) \\ 
\hline
5     & 1.05 (0.18) & 2.23 (0.14) & 18.14 \\ 
7     & 0.93 (0.15) & 2.11 (0.12) & 24.89 \\ 
9     & 0.84 (0.11) & 2.06 (0.09) & 31.49 \\ 
12    & 0.82 (0.22) & 2.03 (0.15) & 40.16 \\ 
\bottomrule[1pt]
\end{tabular}
\caption{U-shape Dataset: MSPE, MAPE, and computation runtime (in seconds) of GS-BART using different numbers of subcolumns (out of 12 total).}
\label{app_Tab:subcol}
\end{table}

\vspace{-5mm}
\baselineskip=14pt 
\singlespacing 
\begingroup
    \setstretch{0.5}
    \bibliographystyle{apalike}
    {\small
    \bibliography{Bibliography-GS-BART}}
\endgroup 

\newpage 

\appendix 

\end{document}